\documentclass[aps,amsmath,amssymb,twocolumn,nofootinbib,superscriptaddress]{revtex4-1}

\usepackage{graphicx}
\usepackage{dcolumn}
\usepackage{bm}
\usepackage{epsfig,hyperref,amsthm}
\usepackage{mathtools}
\usepackage{color}
\usepackage{colordvi}
\usepackage[dvipsnames]{xcolor}
\usepackage{relsize}
\usepackage{accents}
\usepackage{wasysym}  

\newcommand{\ket}[1]{\ensuremath{|#1\rangle}}
\newcommand{\bra}[1]{\ensuremath{\langle #1|}}
\newcommand{\vecP}{{\bf P}}
\newcommand{\vecB}{{\bf B}}
\newcommand{\vecF}{{\bf F}}
\newcommand{\vecA}{{\bm \sigma}}
\newcommand{\Id}{\mathbb{I}}
\newcommand{\asb}{\{\sigma_{a|x}\}_{a,x}}
\newcommand{\proj}[1]{\ket{#1}\!\bra{#1}}

\newcommand{\tr}{\text{tr}}

\newcommand{\kv}{\text{\tiny {\rm KV}}}
\newcommand{\bc}{ \omega_s }
\newcommand{\rIso}{\rho_{\rm{iso}}}
\newcommand{\So}{\mathcal{S}_{\rm O}}

\newcommand{\F}{\mathcal{F}}
\newcommand{\E}{\mathbb{E}}
\newcommand{\Ea}{\mathbb{E}_{\rm A}}
\newcommand{\Eb}{\mathbb{E}_{\rm B}}
\newcommand{\M}{\mathcal{M}}
\newcommand{\W}{\mathcal{W}}

\newcommand{\SW}{\mathcal{S}_{\rm W}}
\newcommand{\SR}{\mathcal{S}_{\rm R}}

\newcommand{\sqp}[1]{\ensuremath{\left[ #1 \right]}}
\newcommand{\smp}[1]{\ensuremath{\left( #1 \right)}}
\newcommand{\bgp}[1]{\ensuremath{\left\lbrace #1 \right\rbrace}}

\newtheorem{theorem}{Theorem}[section]
\newtheorem{lemma}[theorem]{Lemma}
\newtheorem{proposition}[theorem]{Proposition}

\definecolor{nblue}{rgb}{0.2,0.2,0.8}

\begin{document}
\title{Quantum steerability: Characterization, quantification, superactivation, and unbounded amplification}

\author{Chung-Yun Hsieh}
\email{s103022502@m103.nthu.edu.tw}
\affiliation{Department of Physics, National Tsing Hua University, Hsinchu 300, Taiwan}
\author{Yeong-Cherng Liang}
\email{ycliang@mail.ncku.edu.tw}
\affiliation{Department of Physics, National Cheng Kung University, Tainan 701, Taiwan}
\author{Ray-Kuang Lee}
\affiliation{Department of Physics, National Tsing Hua University, Hsinchu 300, Taiwan}
\affiliation{Physics Division, National Center for Theoretical Science, Hsinchu 300, Taiwan}

\date{\today}

\begin{abstract}

Quantum steering, also called Einstein-Podolsky-Rosen steering, is the intriguing phenomenon associated with the ability of spatially separated observers to {\em steer}---by means of local measurements---the set of conditional quantum states accessible by a distant party. In the light of quantum information, {\em all} steerable quantum states are known to be resources for quantum information processing tasks. Here, via a quantity dubbed {\em steering fraction}, we derive a simple, but general criterion that allows one to identify quantum states that can exhibit quantum steering (without having to optimize over the measurements performed by each party), thus making an important step towards the characterization of steerable quantum states. The criterion, in turn, also provides upper bounds on the largest steering-inequality violation achievable by arbitrary finite-dimensional maximally entangled states. For the quantification of steerability, we prove that a strengthened version of the steering fraction is a {\em convex steering monotone} and demonstrate how it is related to two other steering monotones, namely, steerable weight and steering robustness. Using these tools, we further demonstrate the {\em superactivation} of steerability for a well-known family of entangled quantum states, i.e., we show how the steerability of certain entangled, but unsteerable quantum states can be recovered by allowing joint measurements on multiple copies of the same state. In particular, our approach allows one to explicitly construct a steering inequality to manifest this phenomenon.  Finally,  we prove that there exist examples of quantum states (including some which are unsteerable under projective measurements) whose steering-inequality violation can be arbitrarily amplified by allowing joint measurements on as little as three copies of the same state. For completeness, we also demonstrate how the largest steering-inequality violation can be used to bound the largest Bell-inequality violation and derive, analogously, a simple sufficient condition for Bell-nonlocality from the latter.
\end{abstract}


\maketitle

\section{Introduction}
\label{Sec:Intro}
From the famous Einstein-Podolsky-Rosen (EPR) paradox~\cite{EPR} to Bell's seminal discovery~\cite{Bell}, quantum theory has never failed to surprise us with its plethora of intriguing phenomena and mind-boggling applications~\cite{book,RMP-Bell}. Among those who made the bizarre nature of quantum theory evident was Schr\"odinger, who not only coined the term ``entanglement," but also pointed out that quantum theory allows for {\em steering}~\cite{Schrodinger}: through the act of local measurements on one-half of an entangled state, a party can {\em remotely} steer the set of (conditional) quantum states accessible by the other party.

Taking a quantum information perspective, the demonstration of steering can be viewed as the verification of entanglement involving an untrusted party~\cite{Wiseman}. Imagine that two parties Alice and Bob share some quantum state and Alice wants to convince Bob that the shared state is entangled, but Bob does not trust her. If Alice can convince Bob that the shared state indeed exhibits EPR steering, then Bob would be convinced that they share entanglement, as the latter is a prerequisite for steering. Note, however, that shared entanglement is generally insufficient to guarantee steerability. Interestingly, steerability is actually a necessary but generally insufficient condition for the demonstration of Bell nonlocality~\cite{Wiseman,Quintino}. Hence, steering represents a form of quantum inseparability in between entanglement and Bell nonlocality.

Apart from entanglement verification in a partially-trusted scenario, steering has also found applications in the distribution of secret keys in a partially trusted scenario~\cite{steering-QKD}. From a resource perspective, the steerability of a quantum state $\rho$, i.e., whether $\rho$ is steerable and the extent to which it can exhibit steering turns out to provide also an indication for the usefulness of $\rho$ in other quantum information processing tasks. For instance, steerability as quantified by steering robustness~\cite{Piani2015} is monotonically related to the probability of success in the problem of  subchannel discrimination when one is restricted to local measurements aided by one-way communications. 

The characterization of quantum states that are capable of exhibiting steering and the quantification of steerability are thus of relevance not just from a fundamental viewpoint, but also in quantum information. Surprisingly, very little is known in terms of which quantum state is (un)steerable (see, however,~\cite{Wiseman,steering-ineq,Girdhar,Taddei,Bowles:PRA:2016,Werner,Barrett,Almeida}). Here, we derive some generic sufficient conditions for steerability that can be applied to quantum state of arbitrary Hilbert space dimensions. Importantly, in contrast to the conventional approach of steering inequalities~\cite{steering-ineq} where an optimization over the many measurements that can be performed by each party is needed, our criteria only require the relatively straightforward computation of the fully entangled fraction~\cite{Horodecki-1}.

Given that some entangled quantum state $\rho$ cannot exhibit steering~\cite{Werner,Wiseman,Barrett,Almeida,Quintino,Bowles:PRA:2016}, a natural question that arises is whether the steerability of such a state can be {\em superactivated} by allowing joint measurements on multiple copies of $\rho$. In other words, is it possible that some $\rho$ that is not steerable becomes steerable if local measurements are performed instead on $\rho^{\otimes k}$ for some large enough $k$? Building on some recent results established for Bell nonlocality~\cite{Palazuelos,Cavalcanti-PRA}, we provide here an affirmative answer to the above question.

Note that even for a quantum state $\rho$ that is steerable, it is interesting to investigate how its steerability scales with the number of copies. For instance, is it possible to amplify the amount of steering-inequality violation by an {\em arbitrarily large} amount if only a small number of copies are available (see~\cite{Palazuelos} for analogous works in the context of Bell nonlocality)? Again, we provide a positive answer to this question, showing that an unbounded amount of amplification can be obtained by allowing joint measurements on as few as three copies of a quantum state that is barely steerable, or even unsteerable under projective measurements.

The rest of this paper is structured as follows. In Sec.~\ref{Sec:Prelim}, we give a brief overview of some of the basic notions in Bell nonlocality and EPR steering that we will need in subsequent discussions. There, inspired by the work of Cavalcanti {\it et al.}~\cite{Cavalcanti-PRA}, we also introduce the notion of {\em steering fraction} and {\em largest (steering-inequality) violation}, which are crucial quantities that lead to many of the findings mentioned above. For instance, in Sec.~\ref{Sec:Characterization}, we use these quantities to derive (1) a general sufficient condition for an arbitrary quantum state $\rho$ to be steerable and (2) upper bounds on the largest steering-inequality violation of an arbitrary finite-dimensional maximally entangled state as a function of its Hilbert space dimension $d$. Quantification of steerability using a strengthened version of steering fraction is discussed in Sec.~\ref{Sec:QuantifySteering} --- there, we also demonstrate how this steering monotone~\cite{Gallego} is related to the others, such as the steerable weight~\cite{SteerableWeight} and steering robustness~\cite{Piani2015}. In Sec.~\ref{Sec:SuperAmpli}, we show the superactivation of steerablity, provide a procedure to construct a steering inequality for this purpose, and demonstrate unbounded amplification of steerability.
We conclude in Sec.~\ref{Sec:Conclude} with a discussion and some open problems for future research.

\section{Preliminary notions}
\label{Sec:Prelim}

Consider a Bell-type experiment between two parties Alice and Bob. The correlation between measurement outcomes can be succinctly summarized by a vector of joint conditional distributions $\vecP \coloneqq \{P(a,b|x,y)\}$, where $x$ and $a$ ($y$ and $b$) are, respectively, the labels of  Alice's (Bob's) measurement settings and outcomes. The correlation admits a {\it local hidden-variable} (LHV) model if $\vecP$ is Bell-local~\cite{RMP-Bell}, i.e., can be decomposed for all $a,b,x,y$ as 
\begin{equation}\label{Eq:LHV}
	P(a,b|x,y) = \int P_{\lambda} P(a|x,\lambda)P(b|y,\lambda)d\lambda,
\end{equation}
for some {\em fixed} choice of $P_\lambda\ge0$ satisfying $\int P_\lambda\,d\lambda=1$ and single-partite distributions $\{P(a|x,\lambda)\}_{a,x,\lambda}$, and $\{P(b|y,\lambda)\}_{b,y,\lambda}$.

Any correlation that is not Bell-local (henceforth {\em nonlocal}) can be witnessed by the violation of some (linear) Bell inequality,\footnote{Bell inequalities that are not linear in $\vecP$, or which involve {\em complex} combinations of $P(a,b|x,y)$, have also been considered in the literature, but we will not consider them in this paper.}  
\begin{subequations}\label{Eq:BI}
\begin{equation}\label{Eq:BellFunctional}
	\sum_{a,b,x,y} B_{ab|xy}\, P(a,b|x,y)\stackrel{\mbox{\tiny LHV}}{\le} \omega(\vecB),
\end{equation} 
specified by a vector of real numbers $\vecB\coloneqq\{ B_{ab|xy}\}_{a,b,x,y}$ (known as the Bell coefficients) and the {\em local bound}
\begin{align}
\omega(\vecB) &\coloneqq \sup_{{\vecP} \in \text{LHV}} \sum_{a,b,x,y} B_{ab|xy}\, P(a,b|x,y).
\end{align}
\end{subequations}
In the literature, the left-hand side of Eq.~\eqref{Eq:BellFunctional} is also known as a {\em Bell polynomial}~\cite{Werner:PRA:2001} or a {\em Bell functional}~\cite{Buhrman}, as it maps any given correlation $\vecP$ into a real number. 

To determine if a quantum state (and  more generally if a given correlation $\vecP$) is nonlocal, one can, without loss of generality consider Bell coefficients that are non-negative, i.e., $B_{ab|xy}\ge 0$ for all $a,b,x,y$. To see this, it suffices to note that any Bell inequality, Eq.~\eqref{Eq:BellFunctional}, can be cast in a form that involves only non-negative Bell coefficients, e.g., by using the identity $\sum_{a,b} P(a,b|x,y)=1$, which holds for all $x,y$.
 
Specifically, in terms of the {\em nonlocality fraction} $\Gamma$~\cite{Cavalcanti-PRA},
\begin{align}\label{Eq:NonlocalityFraction}
	\Gamma (\vecP,\vecB) &\coloneqq \frac{1}{\omega(\vecB)} \sum_{a,b,x,y} B_{ab|xy}\, P(a,b|x,y),
\end{align}
$\vecP$ violates  the Bell inequality corresponding to $\vecB$ (and hence being nonlocal) if and only if $\Gamma (\vecP, \vecB)>1$.

Importantly, nonlocal quantum correlation 
\begin{equation}\label{Eq:QCor}
	P(a,b|x,y)=\text{tr}\left[\rho\,\left(E_{a|x} \otimes E_{b|y}\right)\,\right]
\end{equation}
can be obtained~\cite{Bell} by performing appropriate local measurements on a certain entangled quantum state $\rho$, where $\Ea\coloneqq\{E_{a|x}\}_{a,x}$ ($\Eb\coloneqq\{E_{b|y}\}_{b,y}$) are the sets of positive-operator-valued measures (POVMs)~\cite{book}  acting on  Alice's (Bob's) Hilbert space. From now onward, we will use $\Ea$ ($\Eb$) to denote a set of POVMs on Alice's (Bob's) Hilbert space, and $\E$ to denote their union, i.e.,  $\E\coloneqq\Ea\cup\Eb$.

Whenever the measurement outcome corresponding to $E_{a|x}$ is observed on Alice's side, quantum theory dictates that the (unnormalized) quantum state 
\begin{equation}\label{Eq:Assemblage}
	\sigma_{a|x}=\text{tr}_\text{A}[\rho\, (E_{a|x} \otimes \Id_\text{B})]
\end{equation}	
is prepared on Bob's side, where $\tr_\text{A}$ denotes the partial trace over Alice's subsystem and $\Id_\text{B}$ is the identity operator acting on Bob's Hilbert space. An {\em assemblage}~\cite{Pusey:PRA:2013} of conditional quantum states $\vecA:=\{\sigma_{a|x}\}_{a,x}$ is said to admit a {\em local-hidden-state} (LHS) model~\cite{Wiseman} if it is {\em unsteerable}, i.e., if it can be decomposed for all $a, x$ as 
\begin{eqnarray}
\sigma_{a|x} = \int P_{\lambda} P(a|x,\lambda)\sigma_{\lambda}\,d\lambda
\end{eqnarray}
for some {\em fixed} choice of $P_\lambda\ge0$ satisfying $\int P_\lambda\,d\lambda=1$, single-partite density matrices $\{\sigma_\lambda\}_\lambda$, and single-partite distribution $\{P(a|x,\lambda)\}_{a,x,\lambda}$. Equivalently, a correlation $\vecP$ admits a LHS model if it can be decomposed as
\begin{equation}\label{Eq:LHS}
	P(a,b|x,y) = \int P_{\lambda} P(a|x,\lambda)\tr\left(E_{b|y}\,\sigma_\lambda \right)d\lambda.
\end{equation}

Conversely, an assemblage $\vecA$ that is steerable can be witnessed by the violation of a steering inequality~\cite{steering-ineq}, 
\begin{subequations}\label{Eq:SteeringIneq}
\begin{equation}\label{Eq:SteeringFunctional}
	\sum_{a,x} \text{tr}(F_{a|x} \sigma_{a|x})\stackrel{\text{\tiny LHS}}{\le} \bc (\vecF),
\end{equation} 
specified by a set of Hermitian matrices $\vecF \coloneqq \{ F_{a|x}\}_{a,x}$
and the {\em steering bound}
\begin{align}\label{Eq:SteeringBound}
	\bc (F) &\coloneqq \sup_{\sigma \in \text{LHS}} \sum_{a,x} \text{tr}(F_{a|x} \sigma_{a|x}).
\end{align}
\end{subequations}
In the literature, the left-hand side of Eq.~\eqref{Eq:SteeringFunctional} is also known as a {\em steering functional}~\cite{JPA15}, as it maps any given assemblage $\vecA=\asb$ to a real number. 

As with Bell nonlocality, in order to determine if a given assemblage is steerable, one can consider, without loss of generality, steering functionals defined only by non-negative, or equivalently positive semidefinite $F_{a|x}$, i.e., $F_{a|x}\succeq 0$ for all $a$, $x$.\footnote{Here and after, the symbol $A\succeq0$ means that the matrix $A$ is positive semidefinite, i.e., having only non-negative eigenvalues.} To see this, it is sufficient to note that any steering inequality, Eq.~\eqref{Eq:SteeringFunctional}, can be rewritten in a form that involves only non-negative $F_{a|x}$, e.g., by using the identity $\sum_a\tr\left(\sigma_{a|x}\right)=1$, which holds for all $x$. 
Hereafter, we thus restrict our attention to $\vecF$ ($\vecB$) having only non-negative $F_{a|x}$ ($B_{ab|xy}$). 
 
In analogy with the nonlocality fraction, we now introduce the  {\em steering fraction}
\begin{align}\label{Eq:SteeringFraction}
	\Gamma_s (\vecA, \vecF)&:= \frac{1}{\bc (\vecF)}\sum_{a,x} \text{tr}(F_{a|x} \sigma_{a|x} ),
\end{align}
to capture the steerability of an assemblage; an assemblage $\vecA$ violates the steering inequality corresponding to $\vecF$ if and only if $\Gamma_s (\vecA, \vecF)>1$. 
Whenever we want to emphasize the steerability of the underlying state $\rho$ giving rise to the assemblage $\rho$, we will write $\Gamma_s(\{\rho,\Ea\},\vecF)$ instead of $\Gamma_s(\vecA,\vecF)$ where $\Ea=\{E_{a|x}\}_{a,x}$, $\vecA:=\asb$, and $\rho$ are understood to satisfy Eq.~\eqref{Eq:Assemblage}.  
In particular, $\rho$ is steerable with $\vecF$ if and only if the largest violation of the steering inequality corresponding to $\vecF$~\cite{JPA15},
\begin{eqnarray}\label{Eq:LV}
	LV_s(\rho,\vecF)\coloneqq\sup_{\Ea}\Gamma_s(\{\rho,\Ea\},\vecF),
\end{eqnarray}
is greater than 1.

As mentioned in Sec.~\ref{Sec:Intro}, Bell nonlocality is a stronger form of quantum nonlocality than quantum steering. 
Let us now illustrate this fact by using the quantities that we have introduced in this section. 
For any given $\vecB=\{B_{ab|xy}\}_{a,b,x,y}$ and Bob's measurements specified by the POVMs $\Eb=\{E_{b|y}\}_{b,y}$, one obtains an {\em induced} steering inequality specified by
\begin{equation}\label{Eq:InducedF}
	\vecF_{(\vecB; \Eb)}\coloneqq\left\{\sum_{b,y} B_{ab|xy} E_{b|y} \right\}_{a,x}.
\end{equation}
Using this equation, the definition
of the steering bound $\bc \sqp{\vecF_{(\vecB; \Eb)}}$, the local bound $\omega(\vecB)$, and the fact that $\tr(\sigma_{\lambda} E_{b|y})$ is only a {\em particular} kind of response function of the form $P(b|y,\lambda)$, one sees that 
\begin{eqnarray}\label{Eq:Bound:LvsS}
	\bc \sqp{\vecF_{(\vecB; \Eb)}}\le \omega (\vecB).
\end{eqnarray}
A geometrical representation of this fact can be found in Fig.~\ref{Fig1}. 

\begin{figure}[h!]
\scalebox{0.8}{\includegraphics{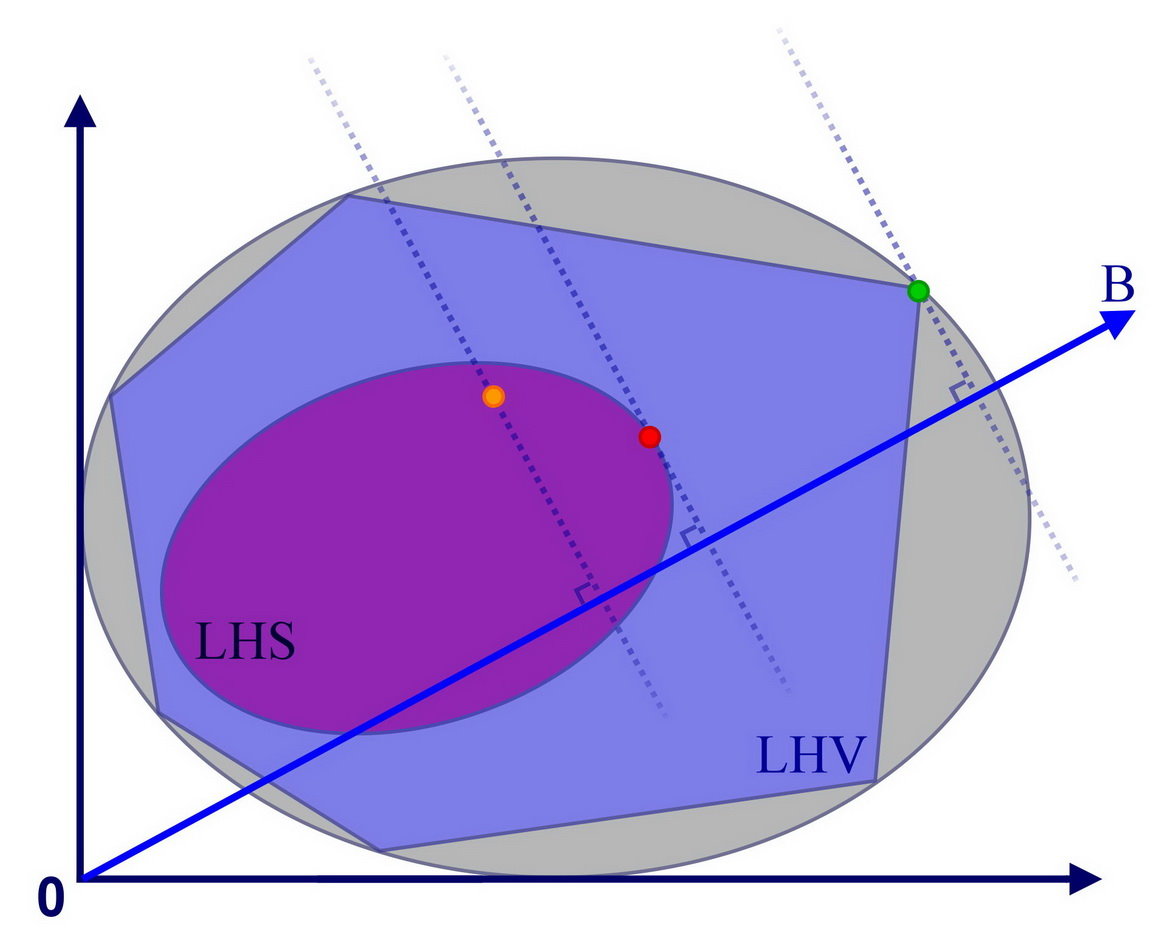} }
\caption{Schematic representation of the relationship between Bell inequalities, the induced steering inequalities, and the different sets of correlations in the space of correlations $\{\vecP\}$. From the outermost to the innermost, we have, respectively, the set of quantum distributions [the grey ellipse, cf. Eq.~\eqref{Eq:QCor}], the set of correlations admitting a LHV model [the blue polygon, cf. Eq.~\eqref{Eq:LHV}], and the set of correlations admitting a LHS model for arbitrary $\Eb$ [the purple ellipse, see Eq.~\eqref{Eq:LHS}]. In this space, the collection of Bell coefficients $\{B_{ab|xy}\}_{a,b,x,y}$ defines a vector $\vecB$ while the Bell functional corresponds to the inner product between $\vecB$ and a generic vector $\vecP$ in this space. Hence, $\omega(\vecB)$ is the largest inner product attainable by all $\vecP\in \text{LHV}$ (blue polygon) achieved, e.g., by the green point, while $\sup_{\Eb}\omega_s \sqp{\vecF_{(\vecB,\Eb )}}$ is the highest inner product attainable by all $\vecP\in$ LHS (purple ellipse) achieved, e.g., by the red point. For {\em any} particular choice of Bob's POVMs the set of correlations defined by Eq.~\eqref{Eq:LHS} form a convex subset (not shown) of the full set of LHS correlations and thus $\omega_s \sqp{\vecF_{(\vecB,\Eb )}}$ may be achieved by some non extremal point of the LHS set. For example, in the figure, $\omega_s [\vecF_{(\vecB,\tilde{\E}_\text{B})}]$ is achieved by the interior orange point for a given $\tilde{\E}_\text{B}$.}
\label{Fig1} 
\end{figure}
Hence, for any correlation $\vecP$ derived by performing local measurements $\Eb$ on Bob's side, and the local measurements $\Ea$ on Alice's side when they share a bipartite state $\rho$, it follows from Eqs.~\eqref{Eq:NonlocalityFraction},~\eqref{Eq:QCor},~\eqref{Eq:SteeringFraction}, and~\eqref{Eq:Bound:LvsS} that
\begin{equation}\label{Eq:BellvsSteering}
	\Gamma(\vecP,\vecB)\le \Gamma_s \sqp{\left\{\rho,\Ea\right\},\vecF_{(\vecB; \Eb)}}.
\end{equation}
From here, it is clear that whenever $\rho$ violates the Bell inequality specified by $\vecB$, i.e., $\Gamma(\vecP,\vecB)>1$, it must also be the case that $\rho$ violates the steering inequality induced by $\vecB$ and $\Eb$, cf. Eq.~\eqref{Eq:InducedF}.

\section{Sufficient condition for steerability and the largest steering-inequality violation}
\label{Sec:Characterization}

Equipped with the tools presented above, our immediate goal now is to derive a sufficient condition for any quantum state $\rho$ acting on $\mathbb{C}^d\otimes\mathbb{C}^d$ to be steerable in terms of its {\em fully entangled fraction} (FEF)~\cite{Horodecki-1,Albererio} 
\begin{equation}\label{Eq:FEF}
\begin{split}
	\F(\rho)\coloneqq &\max_{\Psi}\langle\Psi|\rho|\Psi\rangle\\
	=&\max_{U}\bra{\Psi^+_d}(U\otimes \Id_B)\,\rho\, (U\otimes \Id_B)^\dag\ket{\Psi^+_d},
\end{split}
\end{equation}
which is a quantity closely related to the teleportation~\cite{Teleportation} power of a quantum state. In the above definition,  $\ket{\Psi^+_d}\coloneqq\frac{1}{\sqrt{d}}\sum_{i=0}^{d-1}\ket{i}\ket{i}$ is the generalized singlet, and the maximization is taken over all maximally entangled states $\ket{\Psi}$ in $\mathbb{C}^d\otimes\mathbb{C}^d$, or equivalently over all $d\times d$ unitary operators $U$. Note that in arriving at the second line of Eq.~\eqref{Eq:FEF}, we make use of the fact that $\ket{\Psi^+_d}$ is invariant under local unitary transformation of the form $U\otimes U^*$, where $^*$ denotes complex conjugation. Thus, any maximally entangled state in $\mathbb{C}^d\otimes\mathbb{C}^d$ can be obtained from $\ket{\Psi^+_d}$ by a local unitary transformation acting on Alice's Hilbert space alone. Alternatively, one may also make use of the identity $(A\otimes\mathbb{I}_{\rm B})\ket{\Psi_d^+}=(\mathbb{I}_A\otimes A^\text{\tiny T})\ket{\Psi_d^+}$ which holds for all normal operators $A$~\cite{Jozsa}, where $A^\text{\tiny T}$ is the transpose of $A$ (defined in the Schmidt basis of $\ket{\Psi^+_d}$).

Clearly, the FEF of a state $\rho$ is invariant under local unitary transformation but may decrease when subjected to the $(U\otimes U^*)$-twirling operation~\cite{Horodecki-2,Bennett-96}
\begin{eqnarray}\label{Eq:QuantumTwirling}
	T(\rho)\coloneqq \int_{U(d)}(U \otimes U^*)\rho(U \otimes U^*)^{\dagger}dU,
\end{eqnarray}
where $dU$ is the Haar measure over the group of $d\times d$ unitary matrices $U(d)$. Using a somewhat similar reasoning, one can establish the following lemma (whose proof can be found in Appendix~\ref{App:Proof:Lemma}). 
\begin{lemma}\label{Lemma:Twirling-SF}
For any given state $\rho$, local POVMs $\Ea=\{E_{a|x}\}_{a,x}$, and $\vecF=\{ F_{a|x}\succeq0\}_{a,x}$ acting on $\mathbb{C}^d\otimes\mathbb{C}^d$, there exists another state $\rho'$ and unitary operators $U$ and $U'$ in $U(d)$ such that
\begin{subequations}\label{Eq:Lemma:Conditions}
\begin{gather}
	\F[T({\rho}')]=\bra{\Psi^+_d}\rho'\ket{\Psi^+_d}=\F(\rho')=\F(\rho),\label{Eq:FEFSame}\\		
	\Gamma_s\smp{\bgp{\rho,\tilde{\E}_{\rm A}},\tilde{\vecF}} \ge\Gamma_s[\{T({\rho}'),\Ea\},\vecF], \label{Eq:Gamma_s:Ineq}
\end{gather}
\end{subequations}
where $\tilde{\E}_{\rm A}\coloneqq\left\{U^\dag\,E_{a|x}\,U\right\}_{a,x}$ and $\tilde{\vecF}\coloneqq\left\{U'^\dag\,F_{a|x}\,U'\right\}_{a,x}$.
\end{lemma}
While we shall be concerned, generally, only with non-negative $\vecF$ [we will say $\vecF$ ($\vecB$) is non-negative if it is formed by non-negative $F_{a|x}$ ($B_{ab|xy}$) from now on], it is worth noting that Lemma~\ref{Lemma:Twirling-SF} also holds for $\vecF$ formed by arbitrary Hermitian (but not necessarily non-negative) $F_{a|x}$ if the steering fraction and the corresponding steering bound are defined with an absolute sign, i.e., $\Gamma_s (\vecA, \vecF)= \left|\frac{1}{\bc (\vecF)}\sum_{a,x} \text{tr}(F_{a|x} \sigma_{a|x} )\right|$ and $\bc (\vecF) \coloneqq \sup_{\sigma \in \text{LHS}} \left|\sum_{a,x} \text{tr}(F_{a|x} \sigma_{a|x})\right|$.

Recall from~\cite{Horodecki-2} that the $\left(U\otimes U^*\right)$-twirling operation $T(\rho)$ always gives rise to an isotropic state
\begin{eqnarray}\label{Eq:Isotropic}
	\rho_{\text{iso}}(p)\coloneqq p|\Psi^+_d\rangle\langle\Psi^+_d| + (1-p)\frac{\mathbb{I}}{d^2},
\end{eqnarray}
where $p\in[-\tfrac{1}{d^2-1},1]$ and $\Id$ is the identity operator acting on the composite Hilbert space. In this case, it thus follows from Eq.~\eqref{Eq:Lemma:Conditions} that 
\begin{equation}\label{Eq:Estimate1}
	\Gamma_s \smp{\bgp{\rho, \tilde{\E}_{\rm A}}, \tilde{\vecF}}\ge  \F(\rho) \Gamma_s (\left\{\ket{\Psi_d^+}, \Ea\right\}, \vecF)+Z ,
\end{equation}
where $Z\ge0$ is the contribution of $\Id-\proj{\Psi_d^+}$ towards the steering fraction.
Maximizing both sides over $\E$ and dropping contribution from the second term gives $LV_s(\rho, \vecF)\ge \F(\rho)LV_s(\ket{\Psi_d^+},\vecF)$. Recall from the definition of $LV_s$ that a steering inequality is violated if $LV_s>1$, thus rearranging the term gives the following sufficient condition for $\rho$ to be steerable.
\begin{theorem}\label{Thm:SufficentSteerability}
Given a state $\rho$ and  $\vecF=\{F_{a|x}\succeq 0\}_{a,x}$  acting on $\mathbb{C}^d\otimes\mathbb{C}^d$, a sufficient condition for $\rho$ to be steerable from Alice to Bob\footnote{Note that there exist quantum states that are steerable from Alice to Bob but not the other way around; see, e.g.,~\cite{Bowles:PRL:2015,Quintino}.} is 
\begin{eqnarray}\label{Eq:F-LV}
	\F(\rho)>\frac{1}{LV_s(\ket{\Psi_d^+},\vecF)}.
\end{eqnarray}
\end{theorem}

Since $\F(\rho^{\otimes k})\ge [\F(\rho)]^k$, a direct corollary of Theorem~\ref{Thm:SufficentSteerability} is that, when joint local measurements are allowed, $\rho^{\otimes k}$ becomes steerable if for some $k>1$:
\begin{eqnarray}\label{Eq:F-LV:k-copy}
	\F(\rho)>\left[\frac{1}{LV_s(\ket{\Psi_{d^k}^+},\vecF)}\right]^{\frac{1}{k}}.
\end{eqnarray}

It is worth noting that Theorem~\ref{Thm:SufficentSteerability} holds for general POVMs. If one restricts to projective measurements, it is clear that that the corresponding largest violation, which we denote by $LV^\pi_s(\ket{\Psi_d^+},\vecF)$ may be suboptimal, likewise for the threshold for $\F$ derived from Eq.~\eqref{Eq:F-LV}, i.e., $\F>1/LV^\pi_s(\ket{\Psi_d^+},\vecF) \ge 1/LV_s(\ket{\Psi_d^+},\vecF)$. Using the fact mentioned between Eqs.~\eqref{Eq:FEF} and~\eqref{Eq:QuantumTwirling}, it is easy to see that $\ket{\Psi_d^+}$ in inequality~\eqref{Eq:F-LV} can be replaced by any other state that is local-unitarily equivalent to $\ket{\Psi_d^+}$. Note also that an exactly analogous treatment can be applied to Bell nonlocality, thereby giving a sufficient condition for bipartite Bell nonlocality in terms of the fully entangled fraction of a state $\rho$ (see Appendix~\ref{App:Bell} for details).

Let us also briefly comment on the tightness of the sufficient conditions derived from Theorem~\ref{Thm:SufficentSteerability}. Evidently, in order for the sufficient condition derived therefrom to be tight, the inequality in Eq.~\eqref{Eq:Gamma_s:Ineq} must be saturated and the non-negative term $Z$ that appears in Eq.~\eqref{Eq:Estimate1} must vanish. While the first of these conditions can be met, e.g., by choosing a state $\rho$ that is invariant under the $\left(U\otimes U^*\right)$ twirling (and hence being an isotropic state), the second of these conditions generally cannot  be met at the same time. The relevance of Theorem~\ref{Thm:SufficentSteerability} thus lies in its simplicity and generality, as we now demonstrate with the following explicit examples.

\subsection{Explicit examples of sufficient condition}

As an application of our sufficient condition, consider the $\vecF_\text{\tiny MUB}:=\{\proj{\phi_{a|x}}\}_{a,x}$ induced by a set of $n$ mutually unbiased (orthonormal) bases (MUB) $\{\ket{\phi_{a|x}}\}_{a,x}$ in a $d$-dimensional Hilbert space~\cite{Marciniak15}. $\vecF_\text{\tiny MUB}$ is non-negative since it involves only rank-one projectors~\cite{Marciniak15}. It was shown in~\cite{Marciniak15} that 
\begin{equation}
	LV_s(\ket{\Psi_d^+},\vecF)\ge \max \left\{ \frac{n\sqrt{d}}{n+1+\sqrt{d}}, \frac{d\sqrt{n}}{\sqrt{n}+d-1}\right\}.
\end{equation}
Using this in Theorem~\ref{Thm:SufficentSteerability}, one finds that a sufficient condition for {\em any} bipartite state $\rho$ in $\mathbb{C}^d\otimes \mathbb{C}^d$ to be steerable is 
\begin{equation}\label{Eq:SteerableBound}
	\F(\rho)>\min \left\{ \frac{n+1+\sqrt{d}}{n\sqrt{d}}, \frac{\sqrt{n}+d-1}{d\sqrt{n}}\right\}.
\end{equation}
When $d$ is a power of a prime number~\cite{Marciniak15,Wootters}, one can find $n=d+1$ MUB and the second of the two arguments in the right-hand side of Eq.~\eqref{Eq:SteerableBound} is smaller, thus simplifying the sufficient condition to 
\begin{equation}\label{Eq:F:Sufficient}
	\F(\rho)>\frac{d-1+\sqrt{d+1}}{d\sqrt{d+1}}.
\end{equation}
This implies, for instance, (two-way) steerability of an {\em arbitrary} two-qubit state $\rho$ if $\F(\rho)>\tfrac{1+\sqrt{3}}{2\sqrt{3}}\approx 0.7887$ and an {\em arbitrary} two-qutrit state $\rho$ if $\F(\rho) > \tfrac{2}{3}$, etc. Asymptotically, when $d\to \infty$, the sufficient condition of Eq.~\eqref{Eq:F:Sufficient} simplifies to $\F(\rho)\gtrsim \frac{1}{\sqrt{d}}$, making it evident that this simple criterion becomes especially effective in detecting steerable states for large $d$. Nonetheless, it is worth noting that when $d=2^m$ and with $m\ge 24$, the sufficient condition of Bell nonlocality given in Eq.~\eqref{Eq:SufficientKV} (which is also a sufficient condition for steerability by the fact that any quantum state that is Bell nonlocal is also steerable) already outperforms the sufficient condition given in Eq.~\eqref{Eq:F:Sufficient}.

\subsection{Upper bounds on the largest steering-inequality violation of $\ket{\Psi^+_d}$}
\label{Sec:UB-LV}

Instead of sufficient conditions for steerability, Theorem~\ref{Thm:SufficentSteerability} can also be used to derive upper bounds on the largest steering-inequality violation by $\ket{\Psi^+_d}$ for arbitrary non-negative \vecF, as we now demonstrate. 

Consider again the isotropic state given in Eq.~\eqref{Eq:Isotropic}, which is known to be entangled if and only if $p > p_{\text{ent}}\coloneqq 1/(d+1)$~\cite{Horodecki-2,RMP-Bell}. Moreover,  $\rho_{\text{iso}}(p)$ is non steerable under general POVMs if~\cite{Almeida} $p\le\tilde{p}^{\phi}:=\tfrac{3d-1}{d^2-1}\left(1-\tfrac{1}{d}\right)^d$,  but steerable with projective measurements if and only if~\cite{Wiseman} $p > p_{\text{steer}}\coloneqq (H_d-1)/(d-1)$ where $H_d\coloneqq \sum_{n=1}^{d}\frac{1}{n}$ is the $d$th Harmonic number. For $p\in[0,1]$, it is easy to see that 
\begin{equation}\label{Eq:F-rIso-p}
	\F[\rho_{\text{iso}}(p)]=p+\tfrac{1-p}{d^2}, 
\end{equation}
thus the critical value of FEF beyond which $\rho_{\text{iso}}(p)$ becomes steerable with projective measurements is $\F[\rIso(p)]>\F_{\text{iso,} d}^{\text{steer}}:=\tfrac{H_{d}+H_{d} d-d}{d^2}$. 

In order for this steerability criterion for isotropic state to be compatible with Theorem~\ref{Thm:SufficentSteerability}, we must have $\frac{1}{LV_s^\pi(\ket{\Psi_d^+},\vecF)}\ge \frac{H_{d}+H_{d} d-d}{d^2}$. Otherwise, one would find $\rIso(p)$ with $\F[\rho_{\text{iso}}(p)]<\F_{\text{iso,} d }^{\text{steer}}$ that is steerable according to Theorem~\ref{Thm:SufficentSteerability}, which is a contradiction. Thus, the above necessary and sufficient condition for steerability of $\rho_{\text{iso}}(p)$ with projective measurements implies the following upper bound on the largest steering-inequality violation of $\ket{\Psi^+_d}$.

\begin{theorem}
The largest steering-inequality violation of $\ket{\Psi^+_d}$ for all $\vecF=\{ F_{a|x}\succeq 0\}_{a,x}$ is upper bounded as
\begin{eqnarray}\label{Eq:UpperBound}
	LV_s^\pi(\ket{\Psi_d^+},\vecF)\le \frac{d^2}{H_{d}+H_{d}d-d}
\end{eqnarray}
for projective measurements.
\end{theorem}

To understand the asymptotic behavior of this upper bound, note that when $d\gg 1$ we have
\begin{eqnarray}\label{Eq:AsymptoticFormUpperBound}
	\frac{d}{H_{d}+H_{d}d-d}\approx \frac{1}{H_{d}}< \frac{1}{\ln{d}}.
\end{eqnarray}
This means that $LV_s^\pi(\ket{\Psi_d^+},F)$ scales as $\tfrac{d}{\ln{d}}$ for sufficiently large $d$. In particular, 
it can be shown\footnote{\label{fn:Upperbound}As $d$ increases from 1, the function $\frac{\ln d}{H_{d}+H_{d}d-d}$ increases monotonically until a maximum value at $d=48$ and decreases monotonically after that.} that $\frac{d^2}{H_{d}+H_{d}d-d}\le 1.0900 \tfrac{d}{\ln{d}}$.
Thus our upper bound on  $LV_s^\pi(\ket{\Psi_d^+},\vecF)$ has an asymptotic scaling that improves over the result of Yin {\em at al.}~\cite{JPA15} by a factor of $\tfrac{1}{\ln d}$.\footnote{Their Proposition 2.17 implies an upper bound that scales as $\lesssim d$.}

In addition, by using the sufficient condition of non-steerability of isotropic states under general POVMs [i.e., $\rIso(p)$ is unsteerable under general POVMs if $p\le\tilde{p}^\phi$], one can use Eq.~\eqref{Eq:F-rIso-p} and the same arguments to arrive at the following upper bound under general POVMs:
\begin{eqnarray}\label{Eq:UpperBound_generalPOVM}
LV_s(\ket{\Psi_d^+},F)\le\frac{d^2}{(d^2-1)\tilde{p}^\phi+1}.
\end{eqnarray}
When $d\gg1$~\cite{Almeida}, it can be shown that this upper bound scales as $\frac{ed}{3}$.

Let us also remark that since the upper bound of Eq.~\eqref{Eq:UpperBound} [Eq.~\eqref{Eq:UpperBound_generalPOVM}] holds for {\em all} linear steering inequalities specified by non-negative \vecF, it also serves as a legitimate upper bound on the largest Bell-inequality violation of $\ket{\Psi^+_d}$ with projective measurements (general POVMs) for all {\em linear Bell inequalities} with {\em non-negative} Bell coefficients. A proof of this can be found in Appendix~\ref{App:Bell}. For linear Bell inequalities specified by non-negative \vecB, our upper bound on the largest Bell-inequality violation of $\ket{\Psi^+_d}$ with projective measurements thus has the same scaling as the upper bound due to Palazuelos (see the last equation on page 1971 of \cite{Palazuelos-funct}), but strengthens his by more than a factor of 2. For $d=2$, such an upper bound on the largest Bell-inequality violation of $\ket{\Psi^+_d}$ can be improved further using results from~\cite{Acin,Vertesi}; see Appendix~\ref{App:Bell} for details.

\section{Quantifying Steerability}
\label{Sec:QuantifySteering}

Evidently, as we demonstrate in Sec.~\ref{Sec:Characterization}, steering fraction is a very powerful tool for characterizing the steerability of quantum states. A natural question that arises is whether a maximization of steering fraction over all (non-negative) $\vecF$ leads to a proper {\em steering quantifier}, i.e., a {\em convex steering monotone}~\cite{Gallego}. 

To this end, let us define, for any given assemblage $\vecA=\asb$, the {\em optimal steering fraction} as
\begin{eqnarray}
	\So\left(\vecA\right)\coloneqq\text{max}\left\{ 0,\sup_{\vecF\succeq 0}\Gamma_s\left(\vecA,\vecF\right)-1\right\},
\end{eqnarray}
where the supremum $\sup_{\vecF\succeq 0}$ is taken over all non-negative \vecF. From here, one can further define the optimal steering fraction of a quantum state $\rho$ by optimizing over all assemblages that arise from local measurements on one of the parties.
Superficially, such a quantifier for steerability bears some similarity to that defined in~\cite{Costa}, but, in the steering measure defined therein, there is a further optimization over all possible steering-inequality violations by {\em all possible quantum states}, which is not present in our definition.

In Appendix~\ref{App:SO_Proof}, we prove that $\So$ is indeed a convex steering monotone, i.e., it satisfies the following conditions:
\begin{enumerate}
	\item $\So(\vecA)=0$ for all unsteerable assemblages $\vecA$.
	\item $\So$ does not increase, on average, under deterministic {\em one-way local operations and classical communications} (1W-LOCCs).
	\item For all convex decompositions of $\vecA =\mu \vecA'+(1-\mu) \vecA''$ in terms of other assemblages $\vecA'$ and $\vecA''$ with $0\le\mu\le1$, $\So(\vecA)\le\mu \So(\vecA')+(1-\mu)\So(\vecA'')$.
\end{enumerate}
Moreover, quantitative relations between $\So$ and two other convex steering monotones, namely, steerable weight ($\SW$)~\cite{SteerableWeight}  and steering robustness ($\SR$)~\cite{Piani2015}, can be established, as we now demonstrate.

\subsection{Quantitative relation between optimal steering fraction $\So$ and steerable weight $\SW$}

To begin with, we recall from~\cite{SteerableWeight} that for any assemblage $\vecA=\asb$, $\SW(\vecA)$ is defined as  the minimum non-negative real value $\nu$ satisfying $\sigma_{a|x}=(1-\nu)\sigma_{a|x}^\text{US}+\nu\sigma_{a|x}^\text{S}$ for all $a$ and $x$, where $\vecA^\text{US}\coloneqq\{\sigma_{a|x}^{\text{US}}\}_{a,x}\in \text{LHS}$ and $\vecA^\text{S}\coloneqq\{\sigma_{a|x}^{\text{S}}\}_{a,x}$ is a steerable assemblage. 
In other words, $\SW(\vecA)$ is the minimum weight assigned to a steerable assemblage when optimized over all possible convex decompositions of  $\vecA$ into a steerable assemblage $\vecA^\text{S}$ and an unsteerable assemblage $\vecA^\text{US}$. 
In Appendix~\ref{App:SW}, we establish the following quantitative relations between $\So$ and $\SW$.

\begin{proposition}\label{Prop:So-SW}
  Given an assemblage $\vecA$ with the decomposition $\sigma_{a|x}=[1-\SW(\vecA)]\sigma_{a|x}^\text{\rm US}+\SW(\vecA)\sigma_{a|x}^\text{\rm S}$, where $\vecA^\text{\rm US}\in\text{\rm LHS}$ and $\vecA^\text{\rm S}$ is steerable, we have
\begin{eqnarray}\label{Eq:OSFandSteerableWeight}
	\So(\vecA)\le\SW(\vecA)\So(\vecA^\text{\rm S})\le \So(\vecA)+2\left[1-\SW(\vecA)\right].\quad
\end{eqnarray}
\end{proposition}

Note that if a given assemblage $\vecA$  is steerable, and hence $\So(\vecA^\text{S})\neq 0$, Eq.~\eqref{Eq:OSFandSteerableWeight} can be rearranged to give
\begin{eqnarray}\label{SteerableWeight_and_OSF_Estimate}
	\frac{\So(\vecA)}{\So(\vecA^\text{S})}\le\SW(\vecA)\le\frac{2+\So(\vecA)}{2+\So(\vecA^\text{S})}.
\end{eqnarray}
This means that if $\vecA$ and $\vecA^\text{S}$ are both largely {\em steerable} so that $\So(\vecA), \So(\vecA^\text{S})\gg 1$,\footnote{This happens if there exist $\vecF_1$ and $\vecF_2$ such that $\Gamma_s(\vecA, \vecF_1)\gg 1$ and $\Gamma_s(\vecA^\text{S},\vecF_2)\gg 1$.} Eq.~\eqref{SteerableWeight_and_OSF_Estimate} leads to the following approximation:
\begin{eqnarray}
	\SW(\vecA)\approx\frac{\So(\vecA)}{\So(\vecA^\text{S})},
\end{eqnarray}
which provides an estimate of $\SW(\vecA)$ in terms of $\So(\vecA)$ and $\So(\vecA^\text{S})$ when the two latter quantities are large.

\subsection{Quantitative relation between optimal steering fraction $\So$ and steering robustness $\SR$}

For any given assemblage $\vecA$, its steering robustness $\SR(\vecA)$ is defined as the minimal value $\nu\in[0,\infty)$ such that the convex mixture $\frac{1}{1+\nu}\vecA+\frac{\nu}{1+\nu}\tilde{\vecA}$ is unsteerable for some assemblage $\tilde{\vecA}$. In Appendix~\ref{App:SR}, we derive the following quantitative relations between $\So$ and $\SR$.

\begin{proposition}\label{Prop:So-SR}
For an assemblage $\vecA$ giving the unsteerable decomposition $\vecA^{\rm US}\coloneqq\frac{1}{1+\SR(\vecA)}\vecA+\frac{\SR(\vecA)}{1+\SR(\vecA)}\tilde{\vecA}$ where $\tilde{\vecA}$ is a legitimate assemblage, we have
\begin{equation}\label{Eq:SR-So}
	\SR(\vecA)\So(\tilde{\vecA})-2\le\So(\vecA)\le \SR(\vecA)\left[\So(\tilde{\vecA})+2\right].
\end{equation}
\end{proposition}
Note that if $\tilde{\vecA}$ is steerable, then $\So(\tilde{\vecA})>0$, thus Eq.~\eqref{Eq:SR-So} can be rearranged to give
\begin{eqnarray}\label{Eq:SteeringRobustness_and_OSF_Estimate}
	\frac{\So(\vecA)}{\So(\tilde{\vecA})+2}\le\SR(\vecA)\le\frac{\So(\vecA)+2}{\So(\tilde{\vecA})}.
\end{eqnarray}
As with the case of $\SW$, if $\So(\vecA), \So(\tilde{\vecA})\gg 1$, Eq.~\eqref{Eq:SteeringRobustness_and_OSF_Estimate} implies the following approximation:
\begin{eqnarray}
	\SR(\vecA)\approx\frac{\So(\vecA)}{\So(\tilde{\vecA})},
\end{eqnarray}
which provides an estimate of $\SR(\vecA)$ in terms of the optimal steering fraction $\So$.

\section{Superactivation and Unbounded Amplification of Steerability}
\label{Sec:SuperAmpli}

Let us now turn our attention to the superactivation and amplification of steerability.

\subsection{Superactivation of steerability}
Following the terminology introduced by Palazuelos~\cite{Palazuelos} in the context of Bell nonlocality, we say that the steerability of a quantum state $\rho$ can be {\em superactivated} if it satisfies:
\begin{subequations}\label{Eq:Superactivation}
\begin{gather}
	LV_s(\rho,\vecF)\le 1\quad\forall\,\,\vecF, \\
	\Gamma_s\left(\left\{\rho^{\otimes k},\Ea\right\},\vecF'\right)>1\quad\text{for some $k$, $\Ea$, and $\vecF'$}.
\end{gather}
\end{subequations}
The possibility to superactivate Bell nonlocality---a question originally posed in~\cite{Liang:PRA:2006}---was first demonstrated by Palazuelos~\cite{Palazuelos} using a certain entangled isotropic state.
Their result was soon generalized by Cavalcanti {\em et al.}~\cite{Cavalcanti-PRA} to show that the Bell nonlocality of all entangled states with $\text{FEF}>\tfrac{1}{d}$ can be superactivated. Since all Bell-nonlocal states are also steerable~\cite{Wiseman,Quintino}, while entangled $\rho_\text{iso}(p)$ are exactly those having FEF $>\tfrac{1}{d}$, the steerability of all entangled but unsteerable $\rho_\text{iso}(p)$ can be superactivated (e.g., those with $p_\text{ent}<p\le\tilde{p}^\phi$; see Sec.~\ref{Sec:UB-LV} on page~\pageref{Sec:UB-LV}). 

For the benefit of subsequent discussion on the amplification of steerability, it is worth going through the key steps involved in the proof of this superactivation. To this end, let us first recall from~\cite{KV, Buhrman} the Khot-Vishnoi (KV) nonlocal game, which is parametrized by $\eta\in[0,\tfrac{1}{2}]$.  Let us denote by $G=\{0,1\}^n$ the group of $n$ bit strings with $\oplus$, bitwise addition modulo 2 being the group operation. Consider the (normal) Hadamard subgroup $H$ of $G$ which contains $n$ elements. The cosets of $H$ in $G$ give rise to the quotient group $\tfrac{G}{H}$ with $\tfrac{2^n}{n}$ elements. The KV game can then be written in the form of a Bell inequality, cf. Eq.~\eqref{Eq:BI}, with $\tfrac{2^n}{n}$ settings and $n$ outcomes:\footnote{Evidently, the KV game defined here only makes sense when the number of outputs is a power of 2. Generalization of this to the situation where $n$ can be an arbitrary positive integer has been considered, for example, in~\cite{Palazuelos-funct}.}

\begin{gather}
	\sum_{a\in x,b\in y}\sum_{x,y\in\tfrac{G}{H}}\!\!\! B^{\kv}_{ab|xy}\ P(a,b|x,y) \stackrel{\mbox{\tiny LHV}}{\le} \omega(\vecB^{\kv}),\nonumber\\
	B^{\kv}_{ab|xy}\coloneqq \sum_{g\in G} \frac{n}{2^n} \eta^{w_g}(1-\eta)^{n-w_g}\delta_{a\oplus g,b}\delta_{x\oplus g,y},\label{Eq:BI:KV}
\end{gather}
where $w_g$ is the Hamming weight of $g\in G$  and $\delta_{i,j}$ is the Kronecker delta between $i$ and $j$
\footnote{Note that for $x,y\in\frac{G}{H}$ and $g\in G$, $\delta_{x\oplus g,y} = 1$ if and only if $y$ and $x\oplus g\coloneqq\{h\oplus g | h\in x \}$ are associated with the same coset in the quotient group $\frac{G}{H}$.}
and $\vecB^\kv$ is the set of Bell coefficients defining the KV game.

An important feature of $\vecB^\kv$ given in Eq.~\eqref{Eq:BI:KV} is that $\omega(\vecB^{\kv})\le n^{-\frac{\eta}{1-\eta}}$~\cite{Buhrman,Palazuelos} . For the specific choice~\cite{Palazuelos} of $\eta=\frac{1}{2}-\frac{1}{\ln n}$, which makes sense only for $n\ge 8$, this gives $\omega(\vecB^{\kv})\le n^{-1+\frac{4}{2\ln n}}< \tfrac{C_u}{n}$ with $C_u=e^4$. In this case, performing judiciously chosen rank-1 projective measurements specified by $\E^{\kv}=\Ea^{\kv}\bigcup\Eb^{\kv}$, where $\Ea^{\kv}:=\{E^\kv_{a|x}\}_{a,x}$ and $\Eb^{\kv}:=\{E^\kv_{b|y}\}_{b,y}$, on $\ket{\Psi^+_D}$ (with $D=n$) gives rise to a correlation $\vecP^\kv$ with the following lower bound on the nonlocality fraction:
\begin{equation}\label{Eq:KVestimate}
	\Gamma \left(\vecP^\kv, \vecB^{\kv}\right) > C \frac{D}{(\ln D )^2} 
\end{equation}	
where $C=4e^{-4}$.

Consider now the collection of  non-negative matrices $\vecF^{\kv} \coloneqq \{ F^{\kv}_{a|x}  = \sum_{b,y} B^{\kv}_{ab|xy}E^{\kv}_{b|y}\}_{a,x}$ induced by the KV game and Bob's optimal POVMs $\Eb^\kv$ leading to the lower bound given in Eq.~\eqref{Eq:KVestimate}.  An application of inequality~\eqref{Eq:BellvsSteering} to Eq.~\eqref{Eq:KVestimate}  immediately leads to

\begin{align}
	\Gamma_s \left(\left\{\ket{\Psi_{D}^+}, \Ea^{\kv}\right\}, \vecF^{\kv}\right)&\ge \Gamma \left(\vecP^\kv, \vecB^{\kv}\right) 
	>C \frac{D}{(\ln D )^2} \label{Steering_KV_Estimate}
\end{align}

For any given state $\rho$, Lemma~\ref{Lemma:Twirling-SF}, Eq.~\eqref{Eq:Estimate1}, and Eq.~\eqref{Steering_KV_Estimate} together imply the existence of $\tilde{\E}^{\kv}_{\rm A}\coloneqq\{ \tilde{E}_{a|x}^{\kv}\}_{a,x}$ such that
\begin{eqnarray}\label{Eq:Gamma-s:FEF}
	\Gamma_s \smp{\bgp{\rho, \tilde{\E}^{\kv}_{\rm A}}, \tilde{\vecF}^{\kv}}\ge C\frac{\mathcal{F}(\rho)D}{(\ln D)^2},
\end{eqnarray}
where 
\begin{eqnarray}\label{Eq:ActualSteeringFunctional}
	\tilde{\vecF}^{\kv}\coloneqq\left\{\tilde{F}^{\kv}_{a|x}=\sum_{b,y} {U'}^\dagger B^{\kv}_ {ab|xy}\,  E^{\kv}_{b|y}\, U'\right\}_{a,x}
\end{eqnarray}
again is non-negative.

Note that if $\rho = \rIso(p)^{\otimes k}$, we have both $D=d^k$ and $\F(\rho)=\sqp{\F(\rIso(p))}^k$ scaling exponentially with $k$, while $(\ln{D})^2=k^2(\ln{d})^2$ only increases quadratically with $k$. Thus, Eq.~\eqref{Eq:Gamma-s:FEF} implies that $\F(\rho)>\tfrac{1}{d}$ is a sufficient condition for $\rIso(p)^{\otimes k}$ to be steerable. In other words, for all entangled $\rIso(p)$, $\rIso(p)^{\otimes k}$ is steerable for sufficiently large $k$. In particular, since $\rIso(p)$ for $p_\text{ent}<p\le\tilde{p}^\phi$ is {\em not} single-copy steerable (see Sec.~\ref{Sec:UB-LV} on page~\pageref{Sec:UB-LV}), the steerability of $\rIso(p)$ with $p$ in this interval can be superactivated.

\subsection{Unbounded amplification of steerability}

Given that joint measurements on an appropriately chosen quantum state $\rho$ can lead to the superactivation of steerability, one may ask, as with Bell nonlocality (see~\cite{Palazuelos}), if it is possible to obtain unbounded amplification of steerability of a quantum state with joint measurements. In particular, since it is easier to exhibit EPR steering than Bell nonlocality, can one achieve unbounded violation of steerability (as quantified using steering fraction) using fewer copies of the quantum state? Here, we show that unbounded amplification of steerability can indeed be achieved using as little as three copies of a quantum state, which improves over the result of unbounded amplification for Bell nonlocality due to Palazuelos~\cite{Palazuelos} with five copies. More precisely, our results are summarized in the following Theorem.

\begin{theorem}\label{Thm:Amplification}
For every $\epsilon>0$ and $\delta>0$, there exists an isotropic state $\rIso$ with local dimension $d$ such that
\begin{eqnarray}\label{Eq:Conditions:Amplification}
	LV_s^{\pi}(\rIso,\vecF)\le \epsilon + 1 \quad \& \quad LV_s(\rIso^{\otimes 3},\tilde{\vecF}^{\kv})>\delta
\end{eqnarray}
for all non-negative $\vecF=\{ F_{a|x}\succeq 0\}_{a,x}$, where $\tilde{\vecF}^{\kv}$ is defined in Eq.~\eqref{Eq:ActualSteeringFunctional} with local dimension $d^3$.
Moreover, this $\rIso$ can be chosen to be unsteerable under projective measurements whenever $\epsilon<1$, and steerable whenever $1\le\epsilon$.
\end{theorem}

\begin{proof}
The proof of this is similar to that given by Palazuelos~\cite{Palazuelos} for proving the unbounded amplification of Bell nonlocality using five copies of $\rIso(p)$. First of all, note from Eqs.~\eqref{Eq:LV} and~\eqref{Eq:Isotropic} that
\begin{align}
	LV_s^{\pi}(\rho_{\text{iso}},\vecF)&\le p LV_s^{\pi}(\ket{\Psi_d^+},\vecF)+(1-p)LV_s^{\pi}\left(\frac{\mathbb{I}}{d^2},\vecF\right)\nonumber\\
						  &\le1+p[LV_s^{\pi}(\ket{\Psi_d^+},\vecF)-1].
\end{align}

Recall from Eq.~\eqref{Eq:UpperBound} that $LV_s^{\pi}(\ket{\Psi_d^+},\vecF)$ is upper bounded, thus $LV_s^{\pi}(\rho_{\text{iso}},\vecF)$ is upper bounded by $1+\epsilon$ for any given $\epsilon>0$ if we set
\begin{equation}\label{Eq:p}
	p=\frac{\epsilon}{\frac{d^2}{H_{d}+H_{d}d-d}-1}.
\end{equation}
Evidently, for any given $\epsilon$, we still need to ensure that Eq.~\eqref{Eq:p} indeed gives a legitimate parameter for isotropic states such that $0\le p\le1$. Since the denominator in Eq.~\eqref{Eq:p} is always non-negative and generally increases with $d$ (it approaches $\infty$ as $d\to\infty$), we see that the $p$ defined above is always non-negative and can always be chosen to be upper bounded by 1 for all $\epsilon>0$.
Using this, as well as Eq.~\eqref{Eq:F-rIso-p} in Eq.~\eqref{Eq:Gamma-s:FEF} with $\rho=\rIso(p)^{\otimes k}$ and $\F(\rho)=[\F(\rIso(p))]^k$, we obtain
\begin{align}\label{Eq:Gamma-s:Amplify0}
	\Gamma_s \smp{\bgp{\rho, \tilde{\E}^{\kv}_{\rm A}}, \tilde{\vecF}^{\kv}}\ge 
	\frac{C\,d^k}{(k\ln{d})^2}\left[-\epsilon + \frac{1 + \epsilon}{d^2} + \frac{(d-1) \epsilon}{d - H_d}\right]^k.
\end{align}
For $d\gg1$,
\begin{align}\label{Eq:Gamma-s:Amplify}
	\Gamma_s \smp{\bgp{\rho, \tilde{\E}^{\kv}_{\rm A}}, \tilde{\vecF}^{\kv}}\ge&\frac{C\,d^k\,\epsilon^k}{(k\ln{d})^2}\left( \frac{H_d-1}{d - H_d}\right)^k\nonumber\\
	\approx&\frac{C\,\epsilon^k}{k^2}(\ln{d})^{k-2}.
\end{align}
Thus, for $k\ge3$, we see that $\Gamma_s \left(\left\{\rIso(p)^{\otimes k}, \tilde{\E}^{\kv}_{\rm A}\right\}, \tilde{\vecF}^{\kv}\right)$ with $p$ defined in Eq.~\eqref{Eq:p} can become arbitrarily large if we make $d$ arbitrarily large. In particular, for any given $\epsilon$, $d$ must be large enough so that the $p$ defined in Eq.~\eqref{Eq:p} is larger than $\frac{1}{d+1}$, the critical value of $p$ below which $\rIso(p)$ becomes separable.

Now, a direct comparison between Eq.~\eqref{Eq:p} and the threshold value of $p=p_\text{steer}=(H_d-1)/(d-1)$ [where the isotropic state becomes (un)steerable with projective measurements] shows that 
\begin{equation}
	0<\epsilon\le \kappa_d\coloneqq\frac{(d-H_d)(d+1)(H_d-1)}{(H_d+dH_d-d)(d-1)},
\end{equation}
if and only if the isotropic state with $p$ given in Eq.~\eqref{Eq:p} is unsteerable under projective measurements.
It is easy to verify that (1) the quantity $\kappa_d$ satisfies $\kappa_d<1$ for all $d\ge2$, and (2) $\kappa_d$ rapidly approaches 1 when $d\to\infty$. 
Hence, for every $0<\epsilon<1$, there exists an isotropic state $\rIso$ (with sufficiently large $d$) that is entangled but unsteerable with projective measurements, but which nevertheless attains arbitrarily large steering-inequality violation with $\rIso^{\otimes 3}$.
\end{proof}

Remarks on the implication of Theorem~\ref{Thm:Amplification} are now in order. Firstly, a direct observation shows that $\rIso$ with $p$ given in Eq~\eqref{Eq:p} is always unsteerable under projective measurements if $0\le\epsilon<\kappa_2=0.3$, where one can verify that $\kappa_d$ achieves its minimal value $0.3$ at $d=2$.
This, however, is  still not enough to guarantee that the given isotropic state $\rIso$ is unsteerable under general POVMs due to the lack of exact characterization of steerability under general POVMs.

Second, it is worth noting that the above results also hold if we replace steerability by Bell nonlocality. To see this, let us first remind that the largest Bell-inequality violation under projective measurements is upper  bounded by the upper bound given in Eq.~\eqref{Eq:UpperBound} [see Eq.~\eqref{Eq:UpperBound_LV}]. Next, note that the lower bound on steering fraction that we have presented in Eq.~\eqref{Eq:Gamma-s:Amplify0} actually inherits from a lower bound on the corresponding nonlocality fraction using the KV Bell inequality. Therefore, exactly the same calculation goes through if $\Gamma_s \smp{\bgp{\rho, \tilde{\E}^{\kv}_{\rm A}}, \tilde{\vecF}^{\kv}}$ is replaced by $\Gamma(\vecP,\vecB^\kv)$ with $\vecP$ derived from Eq.~\eqref{Eq:QCor} assuming local POVMs that lead to Eq.~\eqref{Eq:KVestimate}. In other words, for sufficiently large $d$, one can always find entangled isotropic states $\rIso$ that do not violate any Bell inequality with projective measurements, but which nevertheless attain arbitrarily large Bell-inequality violation with $\rIso^{\otimes 3}$. This improves over the result of Palazuelos~\cite{Palazuelos} which requires five copies for unbounded amplification.

\section{Discussion}
\label{Sec:Conclude}

In this paper, we have introduced the tool of steering fraction $\Gamma_s$ and used it to establish novel results spanning across various dimensions of the phenomenon of quantum steering. Below, we briefly summarize these results and comment on some possibilities for future research.

First, we have derived a general sufficient condition for {\em any} bipartite quantum state $\rho$ to be steerable (Bell nonlocal) in terms of its fully entangled fraction, a quantity closely related to the usefulness of $\rho$ for teleportation~\cite{Teleportation}. As we briefly discussed in Sec.~\ref{Sec:Characterization}, we do not expect  these sufficient conditions  to detect all steerable (Bell-nonlocal) states. Nonetheless, let us stress that to determine if a quantum state is steerable (as with determining if a quantum state can exhibit Bell nonlocality; see, e.g.,~\cite{Horodecki,Liang:PRA:2007}) is a notoriously difficult problem, which often requires the optimization over the many parameters used to describe the measurements involved in a steering experiment (and/or the consideration of potentially infinitely many different steering inequalities).

In contrast, the general criterion that we have presented in Theorem~\ref{Thm:SufficentSteerability} for steerability (and Theorem~\ref{Theorem:Suff_Condi_Nonlocality} for Bell nonlocality) only requires a relatively straightforward computation of the fully entangled fraction of the state of interest. Given that these sufficient conditions are likely to be suboptimal, an obvious question that follows is whether one can find an explicit threshold $\F_\text{thr}$ that is smaller than that given by Theorem~\ref{Thm:SufficentSteerability} (Theorem~\ref{Theorem:Suff_Condi_Nonlocality}) such that $\F>\F_\text{thr}$ still guarantees steerability (Bell nonlocality). While this may seem like a difficult problem, recent advances~\cite{algorithm} in the algorithmic construction of local hidden-variable (-state) models may shed some light on this. More generally, it will be interesting to know if our sufficient condition can be strengthened while preserving its computability. In particular, it will be very useful to find analogous steerability (Bell-nonlocality) criteria that are tight.

On the other hand, the aforementioned sufficient condition has also enabled us to derive upper bounds---as functions of $d$---on the largest steering-inequality violation $LV_s$ ($LV_s^\pi$) achievable by the maximally entangled state $\ket{\Psi^+_d}$ under general POVMs (projective measurements). In particular, using the general connection between $LV_s$ and the largest Bell-inequality violation, $LV$, established in Appendix~\ref{App:Bell}, our upper bounds on $LV_s$ and $LV_s^\pi$ imply upper bounds on $LV$ and $LV^\pi$ by $\ket{\Psi^+_d}$ (for non-negative \vecB), respectively. Notably, our upper bound on $LV^\pi$ is somewhat tighter than that due to Palazuelos~\cite{Palazuelos-funct}. If any strengthened sufficient conditions for steerability, as discussed above, are found, it would also be interesting to see if they could lead to tighter (non asymptotic) upper bound(s) on the largest steering-inequality (and/or Bell-inequality) violation attainable by $\ket{\Psi^+_d}$.

The tool of steering fraction $\Gamma_s$, in addition, can be used to quantify steerability. In particular, we showed that if $\Gamma_s$ is optimized over all (non-negative) \vecF, the resulting quantity can be cast as a {\em convex steering monotone}~\cite{Gallego}, which we referred to as the {\em optimal steering fraction} $\So$. We further demonstrated how this monotone is quantitatively related to two other convex steering monotones: steerable weight~\cite{SteerableWeight} and  steering robustness~\cite{Piani2015}. In the light of quantum information, it would be desirable to determine an operational meaning of $\So$, e.g., in the context of some quantum information tasks (cf. steering robustness~\cite{Piani2015}). Establishment of quantitative relations between $\So$ and other convex steering monotones, such as the relative entropy of steering~\cite{Gallego}, would certainly be very welcome. In particular, it would be highly desirable to establish quantitative relations that allow one to estimate $\So$ from other easily computable steering monotones, such as the steerable weight or the steering robustness.

Using the established sufficient condition for steerability, we have also demonstrated the superactivation of steerability, i.e., the phenomenon that certain unsteerable quantum state $\rho$ becomes, for sufficiently large $k$, steerable when joint {\em local} measurements on $\rho^{\otimes k}$ are allowed. 
A drawback of the examples that we have presented here is that they inherit directly from the superactivation of Bell nonlocality due to Palazuelos~\cite{Palazuelos} and Cavalcanti {\em et al.}~\cite{Cavalcanti-PRA}. An obvious question that follows is whether one can construct explicit examples for the superactivation of steerability using quantum states whose Bell nonlocality {\em cannot} be superactivated via joint measurements. 

One the other hand, with joint local measurements,  we showed that the steering-inequality (Bell-inequality) violation of certain barely steerable (Bell-nonlocal) $\rho$ [or even unsteerable (Bell-local) $\rho$ with projective measurements] can be arbitrarily amplified, in particular, giving an arbitrarily large steering-inequality (Bell-inequality) violation with $\rho^{\otimes 3}$. Could such unbounded amplification be achieved using joint measurements on two copies of the same state? Our proof technique, see Eq.~\eqref{Eq:Gamma-s:Amplify}, clearly requires a minimum of three copies for unbounded amplification to take place but it is conceivable that a smaller number of copies suffices if some other steering (Bell) inequality is invoked, a problem that we shall leave for future research.

{\em Note added.} Recently, we became aware of the work of~\cite{Quintino:unpublished} who independently (1) derived a sufficient condition of steerability in terms of the fully entangled fraction and (2) demonstrated the superactivation of steering of the isotropic states.
Moreover, after submission of this work, an anonymous referee of QIP2017 brought to our attention that for any given assemblage, its optimal steering fraction is actually identical to its steering robustness $\SR$, as can be seen from the dual semidefinite programming formulation of steering robustness given in Eq. (41) in~\cite{SDP} (see also~\cite{Piani2015}).

\begin{acknowledgments}
The authors acknowledge useful discussions with Nicolas Brunner, Daniel Cavalcanti, Flavien Hirsch, and Marco T\'ulio Quintino and helpful suggestions from an anonymous referee of AQIS2016. This work is supported by the Ministry of Education, Taiwan, R.O.C., through ``Aiming for the Top University Project" granted to the National Cheng Kung University (NCKU), and by the Ministry of Science and Technology, Taiwan (Grants No. 104-2112-M-006-021-MY3 and No. 105-2628-M-007-003-MY4).
\end{acknowledgments}

\appendix

\section{Proof of Lemma~\ref{Lemma:Twirling-SF}}
\label{App:Proof:Lemma}

Here, we provide a proof of Lemma~\ref{Lemma:Twirling-SF}. 
\begin{proof}

For any given state $\rho'$, local POVMs $\Ea$ and non-negative $\vecF$, let us note that
\begin{align}
&\Gamma_s \left(\left\{T(\rho'),\Ea\right\},\vecF\right)\coloneqq \frac{\sum_{a,x} \text{tr}\left[({E}_{a|x} \otimes {F}_{a|x})T(\rho')\right]}{\bc({\vecF})}\nonumber\\
&=\int_{U(d)} \frac{\sum_{a,x} \text{tr}\left[(U^\dag\,{E}_{a|x}\,U \otimes U^{*\dag}\,{F}_{a|x}\,U^{*}) \rho'\right]}{\bc({\vecF})}\,dU\nonumber\\
&= \int_{U(d)}\Gamma_s\left(\left\{\rho', {\E}_U\right\}, {\vecF}_U\right)dU\nonumber\\
&\le \max_{U\in U(d)}\Gamma_s\left(\left\{\rho', {\E}_U\right\},{\vecF}_U\right),
\label{Eq:MaxGammas:U:Twirl}
\end{align}
where ${\E}_U\coloneqq {\{ U^{\dag}{E}_{a|x}U \}}$ and ${\vecF}_U\coloneqq {\{ U^{*\dag}F_{a|x}U^* \}}$.

Denoting by $U_\Gamma$ the unitary operator achieving the maximum in Eq.~\eqref{Eq:MaxGammas:U:Twirl}, the above inequality implies that
\begin{equation}\label{GammaS:1}
	\Gamma_s\left(\left\{\rho', {\E}_{U_\Gamma}\right\},{\vecF}_{U_\Gamma}\right) \ge \Gamma_s \left(\left\{T(\rho'),{\Ea}\right\},{\vecF}\right).
\end{equation}

For any given state $\rho$, let us further denote by $U_\F$ the unitary operator that maximizes the FEF of $\rho$ in Eq.~\eqref{Eq:FEF}, i.e., 
\begin{equation}
	\F(\rho)=\bra{\Psi^+_d}(U_\F\otimes \Id_B)\,\rho\, (U_\F\otimes \Id_B)^\dag\ket{\Psi^+_d}.
\end{equation}
Defining $\rho':=(U_\F\otimes \Id_B)\,\rho\, (U_\F\otimes \Id_B)^\dag$ and $\tilde{\E}'_{\rm A} = U_\F\, \tilde{\E}_{\rm A}\, U_\F^\dag$, we then have
\begin{subequations}\label{GammaS:2}
\begin{equation}
	\Gamma_s\left(\left\{\rho,\tilde{\E}_{\rm A}\right\},\tilde{\vecF}\right) = \Gamma_s\left(\left\{\rho',\tilde{\E}'_{\rm A}\right\},\tilde{\vecF}\right)
\end{equation}
with 
\begin{equation}
	\F(\rho)=\F(\rho')=\bra{\Psi^+_d}\rho'\ket{\Psi^+_d}=\F\left[T(\rho')\right],
\end{equation}
\end{subequations}
where the last equality follows from the fact~\cite{Horodecki-1,Horodecki-2} that if $\F(\rho')$ is attained with $\ket{\Psi^+_d}$, then $\F\left[T(\rho')\right]$ is attained with $\ket{\Psi^+_d}$.

Combining Eqs.~\eqref{GammaS:1} and~\eqref{GammaS:2} by setting $\E_{U_\Gamma}=\tilde{\E}'_{\rm A}$ and $\vecF_{U_\Gamma}=\tilde{\vecF}$ then gives the desired inequality:
\begin{eqnarray}
	\Gamma_s\left(\left\{\rho,\tilde{\E}_{\rm A}\right\},\tilde{\vecF}\right)&&=\Gamma_s\left(\left\{\rho', {\E}_{U_\Gamma}\right\},{\vecF}_{U_\Gamma}\right) \nonumber\\
	&&\ge \Gamma_s \left(\left\{T(\rho'),{\Ea}\right\},{\vecF}\right),
\end{eqnarray}
with
\begin{equation}
\begin{split}
	\tilde{\E}_{\rm A}=U_\F^\dag\,\tilde{\E}'_{\rm A}\, &U_\F=U_\F^\dag\,U_\Gamma^{\dag}{\Ea}\,U_\Gamma\, U_\F,\\
	\tilde{\vecF}&=U_\F^{*\dag}\, \vecF\, U^*_\F,	
\end{split}
\end{equation}
which completes the proof.
\end{proof}

As a remark, analogous steps but with $\max$ in Eq.~\eqref{Eq:MaxGammas:U:Twirl} replaced by $\min$ leads to the fact that for any given $\rho$, $\Ea$, and non-negative $\vecF$, there exist $\tilde{\E}_{\rm A}$ and $\tilde{\vecF}$ such that $\Gamma_s (\{\rho, \tilde{\E}_{\rm A}\}, \tilde{\vecF}) \le \Gamma_s \left(\left\{T(\rho),\Ea\right\},\vecF\right)$.

\section{Sufficient Condition of Bell nonlocality and Upper Bounds on the Largest Bell-inequality Violation of $\ket{\Psi^+_d}$}
\label{App:Bell}

\subsection{Sufficient condition of Bell nonlocality}

For any given quantum state $\rho$, local POVMs $\E$, and (linear) Bell inequality Eq.~\eqref{Eq:BI}, the largest Bell-inequality violation is defined as~\cite{Palazuelos}: 
\begin{eqnarray}\label{Eq:LV-Bell}
	LV(\rho,\vecB)\coloneqq\sup_{\mathbb{E}}\Gamma\left(\left\{\rho,\mathbb{E}\right\},\vecB\right).
\end{eqnarray}
Using arguments exactly analogous to the proof of Theorem~\ref{Thm:SufficentSteerability}, one can establish the following sufficient condition for a bipartite quantum state $\rho$ to be Bell nonlocal.
\begin{theorem}\label{Theorem:Suff_Condi_Nonlocality}
Given a state $\rho$ acting on $\mathbb{C}^d\otimes\mathbb{C}^d$ and  $\vecB\coloneqq\{B_{ab|xy}\ge 0\}_{a,b,x,y}$, a sufficient condition for $\rho$ to violate the Bell inequality specified by $\vecB$ and hence be Bell nonlocal is
\begin{eqnarray}\label{Eq:Bell:SufficientCondition}
	\F(\rho)>\frac{1}{LV(\ket{\Psi_d^+},B)}.
\end{eqnarray}
\end{theorem}
Note that unlike Theorem~\ref{Thm:SufficentSteerability}, there is no local unitary degree of freedom in the right-hand side of Eq.~\eqref{Eq:Bell:SufficientCondition}. Using again the fact that $\F\left(\rho^{\otimes k}\right)\ge \left[\F(\rho)\right]^k$, a direct corollary of Theorem~\ref{Thm:SufficentSteerability} is the following sufficient condition for  $\rho^{\otimes k}$ to be Bell nonlocal:
\begin{eqnarray}
	\F(\rho)>\left[\frac{1}{LV(\ket{\Psi_d^+},\vecB)}\right]^{\frac{1}{k}}.
\end{eqnarray}

\subsubsection{A sufficient condition based on the Collins-Gisin-Linden-Massar-Popescu-Bell inequality}

As an explicit example, let us consider the family of a two-setting, $d$-outcome inequality due to Collins-Gisin-Linden-Massar-Popescu~\cite{CGLMP}. This inequality can be re-written in a form that involves only the following non-negative coefficients:
\begin{equation}
	B^\text{\tiny CGLMP}_{ab|xy}=    
	\left\{ \begin{array}{c@{\quad\quad}l}
         2+\frac{2(a-b)}{d-1}, & b\ge a\,\,\, \text{and}\,\,\, x+y=0,\\
        \frac{2(a-b-1)}{d-1}, & b<a\,\,\, \text{and}\,\,\, x+y=0,\\
         \frac{2(b-a-1)}{d-1}, & b>a\,\,\, \text{and}\,\,\, x+y=1,\\
        2+\frac{2(b-a)}{d-1}, & b\le a\,\,\, \text{and}\,\,\, x+y=1,\\
         2-\frac{2(b-a-1)}{d-1}, & b>a\,\,\, \text{and}\,\,\, x+y=2,\\
        \frac{2(a-b)}{d-1}, & b\le a\,\,\, \text{and}\,\,\, x+y=2,
        \end{array} \right.
\end{equation}
such that

\begin{equation}\label{Eq:CGLMP}
	\sum_{a,b=0}^{d-1}\sum_{x,y=0,1} B^\text{\tiny CGLMP}_{ab|xy}P(a,b|x,y) \stackrel{\mbox{\tiny LHV}}{\le} 6.
\end{equation}
A (tight) lower bound on the largest Bell-inequality violation of this inequality can be inferred from the result presented in~\cite{CGLMP} as
\begin{eqnarray}\label{Eq:CHSH:LV}
&&LV(\ket{\Psi_d^+},\vecB^\text{\tiny CGLMP}) \nonumber\\
&&=\frac{2}{3}\left\{ 1+d\sum_{k=0}^{\lfloor \frac{d}{2}\rfloor-1}\left(1-\frac{2k}{d-1}\right)\left[q_k-q_{-(k+1)}\right]\right\},\quad
\end{eqnarray}
where $q_k=\tfrac{1}{2d^3\sin^3\left[\pi(k+\frac{1}{4})/d\right]}$.
Putting Eq.~\eqref{Eq:CHSH:LV} into Eq.~\eqref{Eq:Bell:SufficientCondition}, we thus see that 
\begin{equation}\label{Eq:SufficientCGLMP}
		\F(\rho)>\frac{3}{2+2d\sum_{k=0}^{\lfloor \frac{d}{2}\rfloor-1}\left(1-\frac{2k}{d-1}\right)\left[q_k-q_{-(k+1)}\right]}
\end{equation}
is a sufficient condition for Bell nonlocality. For $d=2$, this can be evaluated explicitly to give
\begin{equation}
		\F(\rho)>\frac{3}{2+\sqrt{2}}\approx0.8787,
\end{equation}
whereas, in the asymptotic limit of $d\to\infty$, the sufficient condition becomes $\F(\rho)>0.8611$.

\subsubsection{A sufficient condition based on the Khot-Vishnoi nonlocal game}

Although the sufficient condition of Eq.~\eqref{Eq:SufficientCGLMP} can be applied to an arbitrary Hilbert space dimension $d$, there is no reason to expect that it is optimal for all $d\ge 2$. Indeed, when $d$ is a power of 2 and when $d\ge 2^{10}$, a considerably stronger sufficient condition for Bell nonlocality can be established based on the known lower bound of $LV(\ket{\Psi_d^+},\vecB^\kv)$ given in Eq.~\eqref{Eq:KVestimate}. Explicitly, applying Eq.~\eqref{Eq:KVestimate} to Eq.~\eqref{Eq:Bell:SufficientCondition} gives the following sufficient condition of Bell nonlocality:
\begin{eqnarray}\label{Eq:SufficientKV}
	\F (\rho) > \frac{e^4}{4}\frac{(\ln{d})^2}{d},\quad d=2^m,\quad m\in\mathbb{N},
\end{eqnarray}
where $\mathbb{N}$ is the set of positive integers. This sufficient condition is non trivial (i.e., the lower bound given above is {\em less} than 1) only when $d \ge 2^{10}$. At this critical value of $d$, the sufficient condition of Eq.~\eqref{Eq:SufficientKV}  becomes $\F(\rho)\gtrsim 0.6404$, which is considerably better than that given by  Eq.~\eqref{Eq:SufficientCGLMP}. Notice also that, for $m\ge 3$, the right-hand side of Eq.~\eqref{Eq:SufficientKV} decreases monotonically with increasing $m$. Thus, when measured in terms of the fully entangled fraction, Eq.~\eqref{Eq:FEF}, we see that the fraction of the set of Bell-nonlocal states that can be detected by Eq.~\eqref{Eq:SufficientKV} whenever $d$ is a power of 2  increases monotonically with $d$.

\subsection{Upper bounds on the largest Bell-inequality violation}

In this section, we demonstrate how the largest steering-inequality violation $LV_s$ is related to the largest Bell-inequality violation $LV$ for the case when  $\vecF$ is induced from the given $\vecB$.
\begin{theorem}\label{theorem:LV_Upper_bounded_by_LVs}
Given $\vecB\coloneqq\{B_{ab|xy}\ge 0\}_{a,b,x,y}$ and a state $\rho$ acting on $\mathbb{C}^d\otimes\mathbb{C}^d$, we have
\begin{eqnarray}
	LV(\rho,\vecB)\le \sup_{\Eb}LV_s\left[\rho,\vecF_{(\vecB; \Eb)}\right],
\end{eqnarray}
where the supremum is taken over all possible sets of Bob's POVMs $\Eb:=\{ E_{b|y}\}_{b,y}$, and
$\vecF_{(\vecB;\Eb)}$ is given by Eq.~\eqref{Eq:InducedF}. 
\end{theorem}
\begin{proof}
Let $\Ea:=\{ E_{a|x}\}_{a,x}$ be a generic set of Alice's POVMs and $\E$ be the union of $\Ea$ and $\Eb$.  Suppose that the local POVMs $\Ea$ and $\Eb$ acting on the joint state $\rho$ of Alice and Bob gives rise to the correlation $\vecP$; then Eq.~\eqref{Eq:BellvsSteering} implies that
\begin{equation}
	\Gamma \left(\left\{\rho,\E\right\},\vecB\right)\le \Gamma_s \sqp{\left\{\rho,\Ea\right\},\vecF_{(\vecB;\Eb)}},
\end{equation}
where we write $\Gamma \left(\vecP,\vecB\right):=\Gamma \left(\left\{\rho,\E\right\},\vecB\right)$ to make the dependence of $\vecP$ on $\rho$ and $\E$ explicit.
Thus,
\begin{eqnarray}
&&LV(\rho,B)=\sup_{\E}\Gamma \left(\left\{\rho,\E\right\},\vecB\right)\nonumber\\
&&\le \sup_{\Ea, \Eb} \Gamma_s \sqp{\left\{\rho,\Ea\right\},\vecF_{(\vecB;\Eb)}}\nonumber\\
&&\le \sup_{\Eb} LV_s\left[\rho,\vecF_{(\vecB;\Eb)}\right],
\end{eqnarray}
which completes the proof.
\end{proof}
Theorem \ref{theorem:LV_Upper_bounded_by_LVs} implies that $LV(\rho,\vecB)$ is upper bounded by the highest value of the largest steering-inequality violation of steering inequalities that can be induced by $\vecB$. 
In this sense, we can interpret
$\sup_{\Eb} LV_s[\rho,\vecF_{(\vecB;\Eb)}]$ as the largest steering-inequality violation arising from a given $\vecB$.
In particular, the largest Bell-inequality violation achievable by a maximally entangled state for {\em any} non-negative $\vecB$  under projective measurements must also be {\em upper bounded} by Eq.~\eqref{Eq:UpperBound}:
\begin{eqnarray}\label{Eq:UpperBound_LV}
	LV^\pi(\ket{\Psi_d^+},\vecB)\le \frac{d^2}{H_d+H_d d-d}.
\end{eqnarray}
Note that Eq.~\eqref{Eq:UpperBound_LV} {\em implies}---with the fact given in footnote \ref{fn:Upperbound}---Palazuelos' upper bound~\cite{Palazuelos-funct} of the largest Bell-inequality violation of maximally entangled states under projective measurements (for non-negative $\vecB$).
Also, Eq.~\eqref{Eq:UpperBound_generalPOVM} implies the following upper bound:
\begin{eqnarray}
LV(\ket{\Psi_d^+},F)\le\frac{d^2}{(d^2-1)\tilde{p}^\phi+1},
\end{eqnarray}
which scales as $d$ when $d\gg1$.
It is worth noting that, in the case of general POVMs, Palazuelos' upper bound (Theorem 0.3 in \cite{Palazuelos-funct}) is better than ours by a scaling factor $\frac{1}{\sqrt{d}}$, but we have used a much simpler approach in our derivation (than the operator space theory approach of~\cite{Palazuelos-funct}).

A nice feature of the upper bounds on $LV^{\pi}(\ket{\Psi_d^+},\vecB)$ and $LV(\ket{\Psi_d^+},\vecB)$ presented above is that they apply to all dimensions $d$ and all non-negative $\vecB$. 
The drawback, however, is that they are generally not tight. 
For instance, for the two-qubit maximally entangled state, the inequality above gives $LV^{\pi}(\ket{\Psi_{d=2}^+},\vecB)\le 1.6$, but if we make explicit use of the nonlocal properties of $\rIso(p)$, then this bound can be tightened. 
Firstly, let us recall from~\cite{Acin} that the threshold value of $p$ above which $\rIso(p)$ violates some Bell inequality by projective measurements is given by $p_c=\tfrac{1}{K_G(3)}$, where $K_G(3)$ is Grothendieck's constant of order 3~\cite{Finch}. 
Although the exact value of the constant is not known, it is known to satisfy the following bounds:
\begin{equation}\label{Eq:KG}
	1.4172 \le K_G(3)\le 1.5163,
\end{equation}
where the lower bound is due to V\'ertesi~\cite{Vertesi} and the upper bound is due to Krivine~\cite{Krivine}.

Note that for $d=2$, $p=p_c$ corresponds to a FEF of $\F_c=\tfrac{1}{4}(3p_c+1)$. Then, in order for Eq.~\eqref{Eq:Bell:SufficientCondition} to be consistent with this observation, we must have
\begin{eqnarray}\label{Eq:GrothendickBound}
	LV^{\pi}(\ket{\Psi_2^+},\vecB)\le\frac{1}{\F_c}=\frac{4K_G(3)}{3+K_G(3)}\le 1.2552,
\end{eqnarray}
where both bounds in Eq.~\eqref{Eq:KG} have been used to arrive at the last inequality.

\section{Proofs related to the properties of the optimal steering fraction}
\label{App:So}

\subsection{Proof that $\So$ is a convex steering monotone}\label{App:SO_Proof}

The proof proceeds in two parts. We first show that $\So$ is a convex function that vanishes for unsteerable assemblages. We then show that it is a steering monotone~\cite{Gallego}, that is, non increasing, on average, under one-way local operations and classical communications (1W-LOCCs). The first part of the proof follows from the following lemma.
\begin{lemma}\label{Lemma:So_Convex}
  $\So(\vecA)$ is convex in $\vecA$ and satisfies $\So(\vecA)=0$ $\forall\,\,\vecA\in\text{LHS}$.
\end{lemma}
\begin{proof}
From the definition of $\Gamma_s(\vecA,\vecF)$ and $\So(\vecA)$, it follows immediately that $\So(\vecA)=0$ if $\vecA \in \text{LHS}$. Note that for any {\em convex} decomposition of the assemblage $\sigma_{a|x}=\mu\sigma_{a|x}'+(1-\mu)\sigma_{a|x}''$ with $\mu\in[0,1]$ and $\vecA'$, $\vecA''$ being two legitimate assemblages, we have 
\begin{equation*}
	\sup_{\vecF\succeq 0}\Gamma_s(\vecA,\vecF)\le\mu\sup_{\vecF\succeq 0}\Gamma_s(\vecA',\vecF)+(1-\mu)\sup_{\vecF\succeq 0}\Gamma_s(\vecA'',\vecF).
\end{equation*} 
When $\So(\vecA)=0$, the convexity of $\So(\vecA)$ holds trivially. In the non trivial case when $\So(\vecA)=\sup_{\vecF\succeq 0}\Gamma_s(\vecA,\vecF)-1 > 0$, it follows from the above inequality for the steering fractions that $\So(\vecA)\le\mu \So(\vecA')+(1-\mu)\So(\vecA'')$, which completes the proof of the convexity of $\So$.
\end{proof}

Next, we shall demonstrate the monotonicity of $\So$ by showing the following theorem.
\begin{theorem}\label{Thm:Monotonic}
$\So$ does not increase, on average, under deterministic 1W-LOCCs, i.e.~\cite{Gallego}, 
\begin{equation}
	\sum_\omega P(\omega)\, \So\left(\frac{\M_\omega\left(\vecA\right)}{P(\omega)}\right) \le \So\left(\vecA\right)
\end{equation}
for all assemblages $\vecA=\asb$, where $P(\omega)=\tr\sqp{\M_\omega(\vecA)}$ and $\M_\omega$ is the subchannel of the completely positive map $\M$ labeled by $\omega$, i.e., 
\begin{equation}
	\M_\omega(\cdot)\coloneqq  K_\omega\,\widetilde{\W}_\omega(\cdot)\, K_\omega^\dagger,
\end{equation}
and $\widetilde{\W}_\omega$ is a {\em deterministic wiring map} that transforms a given assemblage $\vecA=\asb$ to another assemblage $\left\{\tilde{\sigma}_{a'|x'}\right\}_{a',x'}$ with setting $x'$ and outcome $a'$:
\begin{equation*}
	[ \widetilde{\W}_\omega(\vecA)]_{x'}\coloneqq \tilde{\sigma}_{a'|x'}=\sum_{a,x}P(x|x',\omega)P(a'|x',a,x,\omega)\sigma_{a|x}.
\end{equation*}
\end{theorem}
To appreciate the motivation of formulating 1W-LOCCs in the above manner and the definition of the trace of the assemblage $\M_\omega(\vecA)$, we refer the readers to~\cite{Gallego}. Moreover, to ease notation, henceforth, we denote $P(x|x',\omega)P(a'|x',a,x,\omega)$ by $Q(a',x',a,x,\omega)$ and define
\begin{eqnarray}\label{Eq:Domega}
	\mathcal{D}_\omega (\vecA)\coloneqq\frac{K_\omega \widetilde{\W}_\omega (\vecA) K_\omega^\dagger}{\tr\sqp{\M_\omega (\vecA)}}.
\end{eqnarray}
To prove Theorem~\ref{Thm:Monotonic}, we shall make use of the following lemma.
\begin{lemma}\label{Lemma:Nonincreasing_Lemma}
For all $\omega$ and assemblages $\vecA$,
\begin{eqnarray}
	\sup_{\vecF\succeq 0} \Gamma_s\sqp{\mathcal{D}_\omega (\vecA),\vecF}\le\sup_{\vecF\succeq 0} \Gamma_s(\vecA,\vecF).
\end{eqnarray}
\end{lemma}
\begin{proof}
From the definitions given in Eqs.~\eqref{Eq:SteeringFraction} and~\eqref{Eq:Domega}, we get
\begin{eqnarray}
	&&\sup_{\vecF\succeq 0} \Gamma_s\sqp{\mathcal{D}_\omega (\vecA),\vecF}\nonumber\\
	&&\coloneqq\sup_{\vecF\succeq 0} \sum_{a',x'}\frac{\text{tr}[F_{a'|x'}K_\omega \sum_{a,x}Q(a',x',a,x,\omega)\sigma_{a|x} 
	K_\omega^\dagger]}{\bc (\vecF) \tr\sqp{\M_\omega (\vecA)}}\nonumber\\
	&&=\sup_{\vecF\succeq 0}\frac{1}{\bc (\vecF)}\sum_{a,x}\text{tr}(\check{F}_{a|x}\sigma_{a|x}),
\end{eqnarray}
where $\check{\vecF}\coloneqq\{\check{F}_{a|x}\succeq 0\}_{a,x}$ is defined by
\begin{eqnarray}
	\check{F}_{a|x}\coloneqq\sum_{a',x'}\frac{Q(a',x',a,x,\omega)K_\omega^\dagger F_{a'|x'} K_\omega}{\tr\sqp{\mathcal{M}_\omega (\vecA)}}.
\end{eqnarray}
Now we have
\begin{eqnarray}
	\bc (\check{\vecF}) :&&=\sup_{\vecA\in\text{LHS}} \sum_{a,x}\text{tr}(\check{F}_{a|x}\sigma_{a|x})\nonumber\\
	&&=\sup_{\vecA\in\text{LHS}} \sum_{a',x'}\text{tr}\left(F_{a'|x'}[\mathcal{D}_\omega (\vecA)]_{a'|x'}\right)\nonumber\\
	&&\le\sup_{\vecA\in\text{LHS}} \sum_{a',x'}\text{tr}(F_{a'|x'}\sigma_{a'|x'})=\bc (\vecF),
\end{eqnarray}
where the last inequality follows from the fact that $\vecA\in\text{LHS}$ implies $\mathcal{D}_\omega (\vecA)\in\text{LHS}$~\cite{Gallego}, and thus $\left\{ \mathcal{D}_\omega(\vecA)\ |\ \vecA\in\text{LHS}\right\}$ is a subset of LHS. Combining the above results, we have $\sup_{\vecF\succeq 0} \Gamma_s\sqp{\mathcal{D}_\omega (\vecA),\vecF}\le\sup_{\vecF\succeq 0} \Gamma_s(\vecA,\check{\vecF})$.
Since we have $\{\check{\vecF}\ |\ \vecF\succeq 0\}\subseteq\{ \vecF\succeq 0\}$, this means $\sup_{\vecF\succeq 0} \Gamma_s(\vecA,\check{\vecF})\le\sup_{\vecF\succeq 0} \Gamma_s(\vecA,\vecF)$, and hence the lemma.
\end{proof}

To complete the proof of Theorem~\ref{Thm:Monotonic}, it suffices to note, first, that when $\So\sqp{\mathcal{D}_\omega (\vecA)}=0$, the inequality $\So\sqp{\mathcal{D}_\omega (\vecA)}\le\So(\vecA)$ holds trivially, whereas when  $\So\sqp{\mathcal{D}_\omega (\vecA)}=\sup_{\vecF\succeq 0} \Gamma_s\sqp{\mathcal{D}_\omega (\vecA),\vecF}-1 > 0$, Lemma~\ref{Lemma:Nonincreasing_Lemma} implies
\begin{eqnarray}
	\So\sqp{\mathcal{D}_\omega (\vecA)}\le\sup_{\vecF\succeq 0} \Gamma_s(\vecA,\vecF)-1\le\So(\vecA).
\end{eqnarray}
This means that $\So\sqp{\mathcal{D}_\omega (\vecA)}\le\So(\vecA)$ in general. Since $\tr\sqp{\M_\omega (\vecA)}\ge 0$ for all $\omega$ and $\sum_\omega \tr\sqp{\M_\omega (\vecA)}\le1$ for (deterministic) 1W-LOCCs, we must have
\begin{eqnarray}\label{LemmaF3_estimate}
	\sum_\omega \tr\sqp{\M_\omega (\vecA)}\So\sqp{\mathcal{D}_\omega (\vecA)}\le\So(\vecA),
\end{eqnarray}
which completes the proof of Theorem~\ref{Thm:Monotonic} by noting that $P(\omega)=\tr\sqp{\M_\omega (\vecA)}$ holds by definition.

\subsection{Proof of quantitative relations between optimal steering fraction $\So$ and steerable weight $\SW$}
\label{App:SW}

Here, we give a proof of Proposition~\ref{Prop:So-SW}.
\begin{proof}
First, note that the chain of inequalities holds trivially if $\vecA$ is unsteerable, since $\SW(\vecA)=0$ in this case. To prove that $\So(\vecA)\le\SW(\vecA)\So(\vecA^\text{S})$ holds in general, we thus assume that $\So(\vecA)=\sup_{\vecF\succeq0} \Gamma_s(\vecA,\vecF)-1 > 0$ and recall from the condition of Proposition~\ref{Prop:So-SW} that $\sigma_{a|x}=[1-\SW(\vecA)]\sigma_{a|x}^\text{\rm US}+\SW(\vecA)\sigma_{a|x}^\text{\rm S}$; then
\begin{align}
	\So(\vecA)=&\sup_{\vecF\succeq 0} \frac{1}{\bc (\vecF)}\sum_{a,x}\text{tr}(F_{a|x}\sigma_{a|x})-1\nonumber\\
	\le&\sup_{\vecF\succeq 0}\frac{1-\SW(\vecA)}{\bc (\vecF)}\sum_{a,x}\text{tr}(F_{a|x}\sigma_{a|x}^\text{US})\nonumber\\
	&+\sup_{\vecF\succeq 0} \frac{\SW(\vecA)}{\bc (\vecF)}\sum_{a,x}\text{tr}(F_{a|x}\sigma_{a|x}^\text{S})-1\nonumber\\
	=&\,\left[1-\SW(\vecA)\right]\left[\sup_{\vecF\succeq 0} \Gamma_s(\vecA^\text{US},\vecF)-1\right]\nonumber\\
	&+\SW(\vecA)\left[\sup_{\vecF\succeq 0} \Gamma_s(\vecA^\text{S},\vecF)-1\right]\nonumber\\
	\le &\, \SW(\vecA)\So(\vecA^\text{S}),
\end{align}
where the last inequality follows from the fact that $\vecA^\text{US}\in\text{LHS}$. This proves the first inequality.

To prove the second inequality, $\SW(\vecA)\So(\vecA^\text{S})\le\So(\vecA)+2\left[1-\SW(\vecA)\right]$, we note that the triangle inequality implies
\begin{equation*}
\begin{split}
\SW(\vecA)\Gamma_s(\vecA^\text{S},\vecF)\le\Gamma_s(\vecA,\vecF)+\left[1-\SW(\vecA)\right]\Gamma_s(\vecA^\text{US},\vecF)
\end{split}
\end{equation*}

Maximizing both sides over all possible $\vecF\succeq 0$ and using the definition of optimal steering fraction gives
\begin{align*}
	\SW(\vecA)\So(\vecA^\text{S})&+\SW(\vecA)\\
	\le \So(\vecA)&+\left[1-\SW(\vecA)\right]\So(\vecA^\text{US})+2-\SW(\vecA)
\end{align*}
Since $\So(\vecA^\text{US})=0$ by definition, simplifying the above inequality therefore leads to the desired inequality and completes the proof.
\end{proof}

\subsection{Proof of quantitative relations between optimal steering fraction $\So$ and steerable robustness $\SR$}
\label{App:SR}

Here, we give a proof of Proposition~\ref{Prop:So-SR}, which proceeds analogously with the proof of Proposition~\ref{Prop:So-SW}.
\begin{proof}
Again, we focus on the nontrivial scenario where $\So(\vecA)=\sup_{\vecF\succeq 0}\Gamma_s(\vecA,\vecF)-1 > 0$. 
To prove $\So(\vecA)\le \SR(\vecA)\left[\So(\tilde{\vecA})+2\right]$, we note from the condition of the theorem $\vecA^{\rm US}=\frac{1}{1+\SR(\vecA)}\vecA+\frac{\SR(\vecA)}{1+\SR(\vecA)}\tilde{\vecA}$, the definitions of $\So$, $\Gamma_s$, and the triangle inequality that
\begin{align*}
\So(\vecA)\le &\, \left[1+ \SR(\vecA)\right]\sup_{\vecF\succeq 0}\Gamma_s(\vecA^\text{US},\vecF)\\
&+\SR(\vecA)\sup_{\vecF\succeq 0}\Gamma_s(\tilde{\vecA},\vecF)-1\\
\le &\, \left[1+\SR(\vecA)\right]\So(\vecA^\text{US})+\SR(\vecA)\So(\tilde{\vecA})+2\SR(\vecA)\\
=&\,\SR(\vecA)\left[\So(\tilde{\vecA})+2\right],
\end{align*}
where the last equality follows from $\vecA^\text{US}\in\text{LHS}$. 

To show the other inequality, we note that
\begin{align*}
\Gamma_s(\vecA,\vecF)+\left[1+\SR(\vecA)\right]\Gamma_s(\vecA^\text{US},\vecF)=\SR(\vecA)\Gamma_s(\tilde{\vecA},\vecF)
\end{align*}
Rearranging the term, taking the supremum over $\vecF\succeq 0$ on both sides, and noting that $\So(\vecA^\text{US})=0$, we thus obtain the desired inequality $\SR(\vecA)\So(\tilde{\vecA})-2\le\So(\vecA).$
\end{proof}

\end{document}